\documentclass[submission]{eptcs}
\providecommand{\event}{} % Name of the event you are submitting to
\usepackage{etoolbox}
\newtoggle{colorON}
\toggletrue{colorON}
\newtoggle{colorElsON}
\toggletrue{colorElsON}
\newtoggle{hidingON}
\toggletrue{hidingON}
\newtoggle{headOutlineON}
\toggletrue{headOutlineON}
\newtoggle{headerOutlineON}
\togglefalse{headOutlineON}

\usepackage{etex}
\usepackage{xcolor}
\usepackage{makeidx}
\usepackage{verbatim}
\usepackage{fancyhdr}
\usepackage{fancybox}
\usepackage{stmaryrd}
\usepackage{amsmath}
\usepackage{mathpartir}
\usepackage{amsthm}
\usepackage{amssymb}
\usepackage{amsfonts}
\usepackage{bm}
\usepackage{amsbsy}
\usepackage[all]{xypic}
\usepackage {indentfirst}
\usepackage{graphicx}
\usepackage[hang]{subfigure}
\usepackage{cite}
\usepackage{proof}
\usepackage{pgf}
\usepackage{tikz}
\usetikzlibrary{positioning,arrows,automata}
\usepackage{cancel}
\usepackage{semantic}
\usepackage{extarrows}
\usepackage{etoolbox} %for defining conditions
 \newtheorem{theorem}{Theorem}%[chapter]
 \newtheorem{proposition}[theorem]{Proposition}
 \newtheorem{lemma}[theorem]{Lemma}
 \newtheorem{corollary}[theorem]{Corollary}

 \theoremstyle{definition}
 \newtheorem{definition}[theorem]{Definition}

 \theoremstyle{remark}

\newcommand{\CCSm}{\ensuremath{\rm{CCS}^{-}}}

\newcommand{\HOCCSm}{\ensuremath{\rm{HOCCS}^{-}}}

\newcommand{\para}{\,|\,}
\newcommand{\n}[1]{\mbox{\rm n($#1$)}} %names
\newcommand{\fpv}[1]{\mbox{\rm fpv($#1$)}} %free process variables
\newcommand{\bpv}[1]{\mbox{\rm bpv($#1$)}} %bound process variables
\newcommand{\pv}[1]{\mbox{\rm pv($#1$)}} %process variables
\newcommand{\hosub}[2]{\{#1/#2\} }
\newcommand{\ve}[1]{\widetilde{#1}}

\newcommand{\st}[1]{\,{\xrightarrow{#1}}\, }
\newcommand{\wt}[1]{{\xLongrightarrow{#1}} }
\newcommand{\rc}[1]{{\color{red} #1}}
\newcommand{\R}{\ensuremath{\mathcal{R}} }

\newcommand{\sepp}{\vspace*{0.4cm}}

\newcommand{\nsepv}[1]{\vspace{0mm}}
\newcommand{\CCSMSB}{\sim_{ccs^{-}} } %strong bisimilarity
\newcommand{\NCCSMSB}{\not\sim_{ccs^{-}} } %neg strong bisimilarity
\newcommand{\CCSMWB}{\approx_{ccs^{-}} } %weak bisimilarity
\newcommand{\NCCSMWB}{\not\approx_{ccs^{-}} } %neg weak bisimilarity
\newcommand{\CCSMQSB}{\sim^q_{ccs^{-}} } %quasi-stong bisimilarity
\newcommand{\CCSMQSBV}{\sim^{qb}_{ccs^{-}} } %quasi-stong bisimilarity
\newcommand{\CCSMBB}{\approx_{ccs^{-}}^{br} } %branching bisimilarity
\newcommand{\HOCCSMSB}{\sim_{hoccs^{-}} } %strong bisimilarity
\newcommand{\HOCCSMWB}{\approx_{hoccs^{-}} } %weak bisimilarity
\newcommand{\HOCCSMQSB}{\sim^q_{hoccs^{-}} } %quasi-stong bisimilarity
\newcommand{\NHOCCSMQSB}{\not\sim^q_{hoccs^{-}} } %neg quasi-stong bisimilarity

\newcommand{\SE}{\equiv }
\newcommand{\DEF}{\stackrel{\textrm{def}}{=}} %definition
\makeatletter
\@ifundefined{newslide}{%

}{%
}
\makeatletter
\@ifundefined{endinput}{%

}{%
}

\newcommand{\nts}[1]{{\textcolor[RGB]{0,221,0}{#1}}}

\newcommand{\ntstrmd}[1]{}
\newcommand{\ntsvv}[1]{{\textcolor[RGB]{0,0,219}{#1}}}

\newcommand{\bctrmd}[1]{{#1}}

\newcommand{\tdup}[1]{}%remove the related comments

\RequirePackage[normalem]{ulem}
\RequirePackage{color}\definecolor{RED}{rgb}{1,0,0}\definecolor{BLUE}{rgb}{0,0,1}

\title{On Bisimulation in Absence of Restriction}
\author{Xian Xu
\thanks{The author is supported by NSF of China (61872142, 62072299) and partially by the project ANR 12IS02001 PACE.}
\institute{East China University of Science and Technology, Shanghai, China (200237)}
\email{xuxian@ecust.edu.cn}
}
\def\titlerunning{}
\def\authorrunning{}

\begin{document}
\maketitle
%\setcounter{page}{1}
%\setlength{\baselineskip}{14pt}

%\oo{TODO : make it a journal paper.
%\begin{itemize}
%\item[DONE] add all the skipped/brief/missing proofs, e.g., Lemma \ref{l:syn-pro-encoding}, Theorem \ref{factor-bigd-smalld} and its preparing lemmas (see ICE2013 paper \cite{Xu13});
%\item[DONE] add the detailed discussion for normal characterization for context bisimulation in $\HOPid$ (see ICE 2013 paper \cite{Xu13});
%\item[DONE] add the discussion of the variant encoding in Section \ref{s:conclusion};
%\item adjust and improve the whole paper to be consistent (including presentation): 
%\begin{itemize}
%\item[DONE] start from Section \ref{s:conclusion}; %, now at Section \ref{smalld-charac}; 
%\item[DONE] %[\bc{to continue from Section 2}] 
%then from the beginning, noticing the result on the variant encoding; 
%\end{itemize}
%\item[DONE] spell-check and maybe more checks;
%\item send.
%\end{itemize}
%}

%\bc{TODO :\rc{make it a (workshop) paper} (technical stuff nearly ok now); } %\bc{$\checkmark$}}
%\begin{itemize}
%\item introduction.  \bc{$\checkmark$}
%\item definition.  \bc{$\checkmark$}
%\item conclusion.  \bc{$\checkmark$}
%\item tidy up (remove the tiny comments by redefining ``$\backslash$tdup").  \bc{$\checkmark$}
%\item page limit and more readable (\rc{flow})  \,\&\,  remove colors and nts %\bc{$\checkmark$}
%\item \greycolor{Words: facilitate, demonstrate, be dedicated to, }
%\end{itemize}
%\sep
%\vspace*{-.5cm}

\iftoggle{headerOutlineON}{%
  % use outline
\nts{TODO: move related discussion materials from repo. ``hp\_amb\_int\_trial" to here (done)}\\
\nts{AND expand further discussion (done).} \\
\nts{AND tidy up to a pper (first-round done).}\\
\nts{AND ~ CHECK THRU again \& tidy up toward a formal paper! (done).}\\
\nts{AND ~ \rc{$\st{~~~~}$}CLEAN UP  (done).}\\
\nts{AND ~ \rc{$\st{~~~~}$}MOVE through AGAIN\& to submit! (done).}\\
\nts{AND ~ \rc{$\st{~~~~}$}prepare submitting! (doing).}\\
} {%
  % no outline
}%toggle hiding

\begin{abstract}
\noindent\emph{\textbf{Abstract}}~ 
We revisit the standard bisimulation equalities in process models free of the restriction operator. 
As is well-known, in general the weak bisimilarity is coarser than the strong bisimilarity because it abstracts from internal actions. In absence of restriction, those internal actions become somewhat visible, so one might wonder if the weak bisimilarity is still `weak'. We show that in both CCScore (i.e., Milner's standard CCS without $\tau$-prefix, summation and relabelling) %(process-passing) 
and its higher-order variant (named HOCCScore), %first-order (name-passing) 
%process models, 
the weak bisimilarity indeed remains weak, i.e., still strictly coarser than the strong bisimilarity, even without the restriction operator. These results can be extended to other first-order or higher-order process models. 
Essentially, this is  due to the direct or indirect existence of the replication operation, which can keep a process retaining its state (i.e., capacity of interaction). By virtue of these observations, we examine a variant of the weak bisimilarity, called quasi-strong bisimilarity.  This quasi-strong bisimilarity  requires the matching of internal actions to be conducted in the strong manner, as for the strong bisimilarity,  and the matching of visible actions to have no trailing internal actions. We exhibit that in CCScore without the restriction operator, the weak bisimilarity exactly collapses onto this quasi-strong bisimilarity, which is moreover shown to coincide with the branching bisimilarity. 
%Moreover
These results reveal that in absence of the restriction operation, some ingredient of the weak bisimilarity indeed turns into strong, particularly the matching of internal actions. 
%It is however not clear if this is the case for higher-order CCS without the restriction, basically due to the indirect definition of the replication operator, which draws on its power inherently from the machinery of process-passing. 

\vspace*{.1cm}
\noindent\emph{keywords}: Strong Bisimulation, Weak Bisimulation, Restriction, Higher-order, First-order, Processes %Parameterization, Encoding, Context bisimulation, Higher-order, First-order, Processes %, Name-passing, Process-passing,

\vspace*{.1cm}
\end{abstract}
\sepp
%\Keywords{Strong Bisimulation, Weak Bisimulation, Restriction, Higher-order, First-order, Process Model, Formal Method}

%----------------------------------------------------------------------------------------------
%% !TEX root = ./main.tex
%introduction

\section{Introduction}\label{s:introduction}

Process models study the behaviour of concurrent systems, particularly their equivalence or degree of similarity \cite{SW01a,San12}. Bisimulation equality, called bisimilarity, is the most exploited notion of such equivalence \cite{Mil89,San11}. For two concurrent systems, a strong bisimulation requires each action, whether external (i.e., visible) or internal (i.e., invisible), of one system to be precisely matched by the other. In contrast, a weak bisimulation, as its name suggests, allows the bisimulation to be observational. Namely, an external action of one system can be matched by the same action of the other, possibly mingled with some internal actions. As is well-known, the weak bisimulation equality is usually coarser than the strong bisimulation equality, because the former can hide some computation inside the system (e.g., implementation). A typical way to achieve such hiding is by the restriction operator, sometimes called the localization operator, which literally has the effect of concealing a port/channel name from being discovered. 
For example, in the language of pure CCS \cite{Mil89}, the following process $M$ has three concurrent  components $P,Q,R$ connected by the operation of parallel composition ($\para$), and the components $P,Q$ share a restricted (or local) name $m$, represented by the restriction operation $(m)$, that can be used to keep $R$ from knowing or sensing something happening between $P$ and $Q$, virtually forming a subsystem.
\[
M \DEF (m)(P\para Q) \para R
\]

% an example? added

Restriction is a frequent and powerful operator in process models. Intensionally it can hide from outside world the critical information, and extensionally it facilitates internal or silent movement (possibly forced synchronization) so that the intrinsic implementation details are transparent to observers. From the viewpoint of computation, it provides a recourse to Turing completeness as well as interactional completeness  in the measurement of computability \cite{SW01a,Fu15}. Concomitantly, some undecidability issues come along with the high computability, e.g., undecidability of bisimulations, making it sometimes succumb to efficient use in practice. So sometimes it is tempting to work without the restriction operator. Although this may bring about certain decrease in expressiveness, the resulting model can avoid being too powerful to be tractable, and in effect be advantageous for practical applications. 
Moreover, fortuitously some models without the restriction operator still turn out to be computational complete, e.g., higher-order processes \cite{LPSS10a,SW01a}. 
In another notable work, Hirschkoff et al. \cite{HP10} study a sub-calculus of the pi-calculus without the restriction and the choice operations but featured with a special top-level replication. The focus of that work is to provide a new congruence result for that sub-calculus, by means of a syntactic characterisation of the (strong) bisimilarity. In this work, by contrast, we are interested in the relationship between the strong and weak bisimilarities in absence of the restriction opertation. 
% that does not include sa new congruence result for the pi-calculus is : bisimilarity is a congruence in the .
% (i.e., Turing complete) \cite{LPSS10a,SW01a}\cite{Others?}. 

It is natural to consider this: if the restriction operator is removed (so one loses the power of hiding), %at least explicitly), 
would weak bisimilarity still be coarser than strong bisimilarity? Or would the `gap' between weak and strong bisimilarities become less pronounced? 
In this work, we look into this question and provide some answer. 
%In particular, we take the standard CCS and Higher-order CCS as the test bed, both without the restriction operator. 
For our purpose, we take CCScore and its higher-order variant (named HOCCScore) as the test bed. %, both without the restriction operator. In particular, for our purpose We denote by 
CCScore (respectively HOCCScore) represents the standard CCS by Milner\cite{Mil89} (respectively plain CHOCS by Thomsen \cite{Tho93}) with the basic first-order (respectively higher-order) concurrency formalism, and without the $\tau$-prefix, summation and relabelling.
In turn, \CCSm\ (respectively \HOCCSm) denotes CCScore (respectively HOCCScore) free of the restriction operation. It is worth noting that \CCSm\ and \HOCCSm\ still admit the replication operation (in the latter, it is a derived operation), otherwise these models would be far less interesting.  
We choose CCScore and HOCCScore because they contain the interesting minimal part of the original models suitable (and non-trivial) for our study. 
Though useful, summation and relabelling are not essential for our work here. %and would (at best) add to unnecessary involvedness. 
In particular, having the $\tau$-prefix would actually defeat the purpose of this work, since it can be deemed as an operation derivable from the restriction operation. To see this, think of the $\tau$-prefix $\tau.P$ as $(c)(c.P\para \overline{c}.0)$ where $(c)Q, c.Q, \overline{c}.Q$ denote the standard CCS restriction, input, and output respectively.

The main goal of this paper is to examine %process models without the restriction operator, particularly 
the bisimulation equalities in calculi \CCSm\ and \HOCCSm. % (called bisimilarity). 
We are interested in the relationship between the weak bisimilarity and the strong bisimilarity, particularly in absence of the restriction operation. %At first glance, this seems quite trivial, because the latter is bound to be strictly subsumed by the former. However, we mention in the passing that this might not always be the case. 
Importantly, the lack of the restriction operator makes the interaction completely exposed to the environment, so that two weakly bisimilar processes may be forced to behave in the manner of strong bisimilarity, from certain perspective. Moreover, the situation can further vary when it comes to the style of interaction, say first-order synchronization (i.e., implicit name-passing) or higher-order process-passing.   

% We demonstrate that in some first-order process models, weak bisimilarity does collapses onto strong bisimilarity. In contrast, in some higher-order models weak bisimilarity stay put, i.e., strictly coarser than strong bisimilarity. This gives evidence to and lend confidence to the following conjectural statement:
% \[
% \mbox{\emph{In first-order process models, weak bisimilarity is coincident with strong bisimilarity.}}
% \]

\paragraph*{\textbf{Contribution}} ~ 
We demonstrate that in \CCSm, the weak bisimilarity does somehow collapse onto a bisimulation equality called quasi-strong bisimilarity. In particular, this quasi-strong bisimilarity requests strong bisimulation on internal moves, and almost the same as the weak bisimilarity does for external actions. By `almost', we mean that in the matching of an external (i.e., visible) action $\alpha$, one can make a few internal ($\tau$) actions before $\alpha$ but none after it, i.e., $\wt{}\st{\alpha}$ in standard notation  (also known as the `delay' transition \cite{SW01a}). The quasi-strong bisimilarity, as it appears, strengthens the weak bisimilarity, and moves toward the strong bisimilarity. 
Moreover, as a corollary, we show that the quasi-strong bisimilarity actually concides with the branching bisimilarity, which in turn implies the coincidence between the weak bisimilarity and the branching bisimilarity.

Specifically, we prove in detail that in \CCSm, the weak bisimilarity can indeed be tightened, to be coincident with the quasi-strong bisimilarity. That being said, there appears to be still some distance from the strong bisimilarity. %, i.e., the quasi-strong bisimilarity or weak bisimilarity is still coarser than the strong bisimilarity. 
This is essentially attributed to the replication operator, which can introduce infinity that in turn generates some kind of state-preserving behaviour (if a process maintains all its interactional capability after doing some action, we refer to that action as state-preserving; otherwise it is state-changing). Consequently to some extent, replication %(or recursion) 
plays a role of realizing immutability in a concurrent system even without the restriction operator.

In contrast, in \HOCCSm\ %higher-order CCS without the restriction operator, 
the weak bisimilarity appears resilient to the removal of the restriction and stays put anyhow. That is, the weak bisimilarity retains being strictly coarser than the strong bisimilarity, and what is more, we do not know how to strengthen it to the quasi-strong bisimilarity as for \CCSm. This intrinsically develops from the complexity of process-passing, which enables one to encode recursion and yield richer behaviours subsequently. In that sense, erasing the restriction operator turns out to have little effect on the behavioural equivalences.  

This work can potentially help to precisely identify the boundaries between different notions of bisimulations, while clarifying further the features of the considered process models and the properties of the processes that are classified (differently) by the different bisimulations.
In the theoretical regard, 
the results of this work give evidence and %hopefully 
lend confidence to related study of %other first-order or higher-order process models 
concurrent models 
concerning the relationship between the weak and strong bisimilarities in a setting free of the restriction operator.
% In regard to application, the results of this work can hopefully help in choosing a suitable bisimulation (for instance, the less demanding one with clauses easy to handle or minimizing the matching bahevior) when it comes to concrete application scenarios.
In regard to application, the results of this work can hopefully help in choosing a suitable bisimulation (for instance, a less demanding one with clauses easy to handle) when it comes to practical scenarios.
Also the technical arguments for these results might be of independent interest.

%\nts{TODO from HERE...; see nt in the beginning.}

%\nts{TODO MORE? intro... maybe (not)?:!...}

\paragraph*{\textbf{Organization}} ~ 
The remainder of the paper is organized as follows. %...\nts{TO ADD} 
Section \ref{s:hoccsm} tackles the bisimilarity in \HOCCSm. %\higher-order CCS without the restriction operator. 
Section \ref{s:ccsm} deals with the \CCSm\ situation. In both sections, we first define the syntax and semantics of the corresponding calculus, and then present the main results with detailed discussion.  Section \ref{s:conclusion} concludes the paper and points to some future work.

%We mention in the passing the technical stuff that might be of independent interest.

%---------------------------
% Local Variables:
% mode: LaTeX
% TeX-master: "main.tex"
% End:

%% !TEX root = ./main.tex
%

% ...

% \sepp\sepp
% \fbox{\nts{\xxa{RE-ENGAGE : The following sections ~$\longrightarrow$ ~~ %(0) ADD the missing parts (e.g., Defs); 
% (1) TIDY UP! FOLLOW the $\clubsuit\clubsuit$}}}
% \sepp\sepp

% %\section{Noting down to the earth...}\label{s:introduction}

% ...

%\section{On bisimilarity in \HOCCSm}\label{s:hoccsm}
\section{On the bisimulation equality in \HOCCSm}\label{s:hoccsm}

%\HOCCSm\ stands for Higher-Order CCS without the restriction operator. 
We first define \HOCCSm, i.e., HOCCScore without the restriction operation, and then discuss the relationship between the strong and weak bisimulation equalities, that is, the bisimilarities.

\subsection{Calculus \HOCCSm}
%Definition of the higher-order CCS \HOCCSm. \nts{TO ADD: ADAPTP from related work (maybe that in ``hp\_comcan\_gl\_p1''; see ``hoccs\_ref.zip'' here in this repo.)}
%\sepp\sepp

%-----------------------------------------------------------------------

%We define the higher-order CCS without the restriction operator, notation \HOCCSm.
A \HOCCSm\ process is given by the following grammar. 
We denote names by lowercase letters, processes by uppercase letters, and process variables by $X,Y,Z$.
\[
\begin{array}{l}
P,P' ::= 0 \;\Big{|}\; X \;\Big{|}\;  a(X).P \;\Big{|}\; \overline{a}P'.P \;\Big{|}\; P\para P' 
%\,\Big{|}\, (c)T \,\Big{|}\,  !u(X).T \,\Big{|}\, !\overline{u}T'.T
\end{array}
\]

The operators have their standard meaning: input prefix $a(X).P$, output prefix $\overline{a}P'.P$, and parallel composition $P\para P'$.
%; restriction: $(c)T$ in which $c$ is bound. 
We stipulate parallel composition to have the least precedence.

A process variable $X$ occurring in $P$ is bound by input-prefix $a(X).P$ and free otherwise. 
We use $\fpv{\cdot}$, $\bpv{\cdot}$, $\pv{\cdot}$ respectively to denote free process variables, bound process variables and process variables in a set of processes. Additionally, we use $\n{\cdot}$ to denote the names in a set of processes. A name or process variable is fresh if it does not appear in the processes under consideration. 
Closed processes are those having no free variables and considered by default in discussion.
As usual, we use $a.0$ as a shortcut for $a(X).0$, and $\overline{a}.0$ for $\overline{a}0.0$; moreover, the trailing $0$ is often omitted. 
Sometimes for clarity, we may write %$\overline{u}(A)$ or 
$\overline{a}[A]$ for higher-order output.
A tilde $\ve{\cdot}$ represents a tuple.
A higher-order substitution $P\hosub{A}{X}$ replaces free occurrences of variable $X$ with $A$ and can be extended to tuples in the expected entry-wise way.

A context $C$, or $C[\cdot]$ to emphasize the hole in it, is a process with some subprocess replaced by the hole $[\cdot]$, and $C[A]$ is the process obtained by substituting $A$ for the hole.
%In particular, contexts can be extended to multihole ones in the standard way \cite{SW01a}; basically a multihole context has several holes each of which may occur a couple of times. By default, we consider one-hole contexts.
We denote by $E[\ve{X}]$ the process expression $E$ (possibly) with free occurrence of the variables $\ve{X}$, and $E[\ve{A}]$ stands for $E\hosub{\ve{A}}{\ve{X}}$. % but avoiding name-capture. 
Essentially, $E[X]$ can be treated as a multihole context \cite{SW01a}, and sometimes we also write $E[\cdot]$ in the discussion.

%A context $C[\cdot]$ and an $E[X]$ are different in some sort. Essentially, $E[X]$ is a kind of multihole context allowing multiple occurrences of the hole, and disallows name capture (whereas $C[\cdot]$ ignores name capture) when there is the restriction operator. Here since we are working without the restriction operator, $E[X]$ can be treated simply as a multihole context (see \cite{SW01a} for more details about contexts). So we sometimes also write $E[\cdot]$ in the discussion. %, bearing in mind its difference from a general context. 

The semantics of \HOCCSm\ (on closed processes) 
is as below. The symmetric rules are omitted.
\[
\begin{array}{ll}
\frac{\displaystyle }{\displaystyle a(X).P\st{a(A)} P\hosub{A}{X}}   & \qquad
\frac{}{\displaystyle \overline{a}A.P\st{\overline{a}A} P}  \\\\
\frac{\displaystyle P\st{\lambda} P'}{\displaystyle P\para Q\st{\lambda} P'\para Q}
& \qquad 
\frac{\displaystyle P\st{a(A)} P'\quad  Q\st{\overline{a}A} Q'}{\displaystyle P\para Q \st{\tau} P'\,|\,Q'}\; 
\end{array}
\]

% \[
% \begin{array}{ll}
% \frac{\displaystyle }{\displaystyle a(X).T\st{a(A)} T\hosub{A}{X}}  &
% \frac{}{\displaystyle \overline{a}A.T\st{\overline{a}A} T}  \\
% \frac{}{\displaystyle !\overline{a}A.T\st{\overline{a}A} T \para !\overline{a}A.T} &
% \frac{\displaystyle }{\displaystyle !a(X).T\st{a(A)} T\hosub{A}{X}\para !a(X).T } \\
% \frac{\displaystyle T\st{\lambda} T'}{\displaystyle (c)T\st{\lambda} (c)T'}{\scriptstyle c\not\in \n{\lambda}} \qquad &
% \frac{\displaystyle T\st{(\ve{c})\overline{a}A} T'}{\displaystyle (d)T\st{(d)(\ve{c})\overline{a}A} T'} {\scriptstyle d \in \fn{A}{-}\{\ve{c},a\}} \\
% \frac{\displaystyle T\st{\lambda} T'}{\displaystyle T\para T_1\st{\lambda} T'\para T_1} {\scriptstyle \bn{\lambda}\,\cap\, \fn{T_1}=\emptyset} \qquad &
% \frac{\displaystyle T_1\st{a(A)} T_1'\quad  T_2\st{(\ve{c})\overline{a}A} T_2'}{\displaystyle T_1\para T_2 \st{\tau}(\ve{c})(T_1'\,|\,T_2')}\;  \xxxx{\scriptstyle \ve{c}\,\cap\, \fn{T_1} = \emptyset} %&
% \end{array}
% \]

We denote by $\alpha,\lambda$ the actions: internal move ($\tau$), input ($a(A)$), output ($\overline{a}A$).
% in which $\ve{c}$ is some local names carried by $A$ during the output. %We always assume no name capture with resort to $\alpha$-conversion.
Operations $\fpv{\cdot}$, $\bpv{\cdot}$, $\pv{\cdot}$, $\n{\cdot}$ can be similarly defined on actions.
As usual, $\wt{}$ is the reflexive transitive closure of internal actions, % $\tau$, 
and $\wt{\lambda}$ is $\wt{}\xrightarrow{\lambda}\wt{}$. Also $\wt{\hat{\lambda}}$ is $\wt{}$ when $\lambda$ is $\tau$ and $\wt{\lambda}$ otherwise. We use $\st{\tau}_k$ to mean $k$ consecutive $\tau$'s.
%The notations $\wt{}$, $\wt{\lambda}$ and $\wt{\widehat{\lambda}}$ are similar to those in \FOPi. %We also reuse $\equiv$ for the structural congruence in \HOPi\ (and also \HOPiDd\ to be defined shortly) \cite{SW01a}, and this shall not raise confusion under specific context.
%``$P\wt{\hat{\lambda}} P'$" abbreviates ``there is a process $P'$ such that $P\wt{\hat{\lambda}} P'$".
 %$P\wt{}\cdot \mathcal{R}\, Q$ means $P\wt{} Q'$ and $Q' \,\mathcal{R}\, Q$ (i.e. $(Q',Q)\in \mathcal{R}$), where $\mathcal{R}$ is a binary relation.
%We say relation $\mathcal{R}$ is closed under (variable) substitution if $(E\hosub{A}{X},F\hosub{A}{X})\in \mathcal{R}$ for any $A$ whenever $(E,F)\in \mathcal{R}$.
%A process diverges if it can perform an infinite $\tau$ sequence.
For a binary relation \R, we use $P\,\R\, Q$ as a shortcut for $(P,Q)\in \R$. 
Sometimes we write $P\st{\lambda}\,\mathcal{R}\, P'$ to mean that there exists $P''$ such that $P\st{\lambda} P''$ and $P''\,\mathcal{R}\, P'$.

We denote by $\SE$ the standard structural congruence \cite{MPW92,SW01a}. It is the smallest equivalence relation satisfying $\alpha$-convertibility over (bound) process variables, the monoid laws and commutative laws for parallel composition. That is,
\[
\begin{array}{ll}
a(X).P \SE a(Y).P\hosub{Y}{X} ~~(Y \mbox{ fresh}) &\qquad  P\para 0 \SE P \\\\
P\para (Q\para R) \SE (P\para Q)\para R &\qquad  P\para Q \SE Q\para P
\end{array}
\]

%,  for both composition 
%and restriction, and a distributive law $(c)(P\para Q) \SE (c)P\para Q$ (if $c\notin \fn{Q}$).

%We work up-to $\alpha$-conversion and always assume no capture.
As is well-known, replication can be derived in \HOCCSm\cite{Tho93,SW01a,LPSS08}. That is, we can define 
\[
!P \DEF \overline{c}Q_{\scriptscriptstyle c,P} \para Q_{\scriptscriptstyle c,P}, \qquad\qquad Q_{\scriptscriptstyle c,P}\DEF c(X).(\overline{c}X \para X \para P) \qquad\qquad(c \mbox{ fresh})
\] Sometimes the name $c$ used to achieve such a replication is referred to as a replicator name. 
Further, we can also define the so-called guarded replication for a prefix $\phi$ (input or output).
% \[
% !\phi.P \DEF Q_{\scriptscriptstyle c,\phi,P} \para \overline{c}Q_{\scriptscriptstyle c,\phi,P}, \qquad\qquad 
% Q_{\scriptscriptstyle c,\phi,P} \DEF c(X).(\phi.(X\para P \para \overline{c}X))
% \]
\[
\begin{array}{l}
!^{\scriptscriptstyle g}\, \phi.P \DEF \overline{c}Q_{\scriptscriptstyle c,\phi,P} \para Q_{\scriptscriptstyle c,\phi,P}, \\\\ %\qquad\qquad 
Q_{\scriptscriptstyle c,\phi,P} \DEF c(X).(\phi.(\overline{c}X\para X\para P))
\end{array}
\]

% \nts{\Large TO ADAPT the folloiwng (in small font)... until double line ...}
% \sepp\sepp

% {\small

\subsection*{\textit{Context bisimulation}}

Throughout, we are relying on the following standard notion of context bisimulation \cite{San92,San94}. 
\begin{definition}[Context bisimulation]\label{context-bisimulation}
A symmetric relation $\mathcal{R}$ on (closed) \HOCCSm\ processes is a (weak) context bisimulation (respectively strong context bisimulation), if $P\,\mathcal{R}\, Q$ implies the following properties:
\begin{enumerate}
\item if $P \st{\alpha} P'$ in which $\alpha$ is $a(A)$ or $\tau$, then $Q \wt{\widehat{\alpha}} Q'$ (respectively $Q\st{\alpha} Q'$) for some $Q'$ and $P'\,\mathcal{R}\, Q'$.
% \item if $P \st{(\ve{c})\overline{a}A} P'$, then $Q \wt{(\ve{d})\overline{a}B} Q'$ for some $B$ that is of the same type as $A$, i.e., a process abstraction, a name abstraction or not an abstraction, and moreover for every $E[X]$ such that $\{\ve{c},\ve{d}\}\cap \fn{E}=\emptyset$ it holds that $(\ve{c})(E[A]\para P') \; \mathcal{R}\;  (\ve{d})(E[B]\para Q')$.
% \item if $P \st{(\ve{c})\overline{a}A} P'$ in which $A$ is a process abstraction, a name abstraction or not an abstraction, then $Q \wt{(\ve{d})\overline{a}B} Q'$ for some $B$ that is respectively a process abstraction, a name abstraction or not an abstraction, and moreover for every $E[X]$ such that $\{\ve{c},\ve{d}\}\cap \fn{E}=\emptyset$ it holds that 
% \begin{equation}
% (\ve{c})(E[A]\para P') \; \mathcal{R}\;  (\ve{d})(E[B]\para Q'). \nonumber\tag{*}%\eqnumber{\triangle}
% \end{equation}
\item if $P \st{\overline{a}A} P'$, then $Q \wt{\overline{a}B} Q'$ (respectively $Q \st{\overline{a}B} Q'$) for some $B$ and $Q'$, and for every $E[X]$ it holds that 
\begin{equation}
E[A]\para P' \; \mathcal{R}\;  E[B]\para Q' \nonumber\tag{*}%\eqnumber{\triangle}
\end{equation}
\end{enumerate}
The (weak) context bisimilarity (respectively strong context bisimilarity), denoted by $\HOCCSMWB$ (respectively $\HOCCSMSB$), is the largest context bisimulation (respectively strong context bisimulation).
\end{definition}

%Relation $\SCB$ denotes the strong context bisimilarity. 
It is well-known that both $\HOCCSMWB$ and $\HOCCSMSB$ are congruences \cite{San92,Tho93,San94,SW01a}. 
Relation $\HOCCSMWB$ (similar for $\HOCCSMSB$) can be extended to open processes in the usual way: suppose %$\ve{X}\subseteq \fpv{P,P'}$,
$\ve{X} = \fpv{P,P'}$, 
then $P\HOCCSMWB P'$ if and only if $P\hosub{\ve{A}}{\ve{X}} \HOCCSMWB P'\hosub{\ve{A}}{\ve{X}}$ for all closed $\ve{A}$. 
\vspace*{.1cm}

It should be clear that by the definitions, the following implications are true (see \cite{San92,SW01a}). %The implications of Lemma \ref{l:hoccsm_bisi_rels} are immediate by the definitions.
\begin{lemma}\label{l:hoccsm_bisi_rels}
It holds that $\SE \;\subseteq\; \HOCCSMSB \;\subseteq\; \HOCCSMWB$.
\end{lemma}

\subsection{The weak context bisimilarity is still weaker than strong bisimilarity in \HOCCSm}

%We denote by $\HOCCSMSB$ the strong context bisimilarity in \HOCCSm. 
The following lemma states that the weak context bisimilarity is strictly coarser than the strong context bisimilarity.
\begin{lemma}\label{l:hopi_bisi_ws_coin}
On (closed) \HOCCSm\ process we have $\HOCCSMSB \,\subsetneq\, \HOCCSMWB$. 
\end{lemma}
\begin{proof}
To see that in \HOCCSm, the weak bisimilarity is strictly coarser than the strong bisimilarity, i.e., $\HOCCSMWB \,\not\subseteq\, \HOCCSMSB$, we examine the `replication', which is reproduced below for convenience.
\[
!P \DEF \overline{c}Q_{\scriptscriptstyle c,P} \para Q_{\scriptscriptstyle c,P}, \qquad\qquad Q_{\scriptscriptstyle c,P}\DEF c(X).(\overline{c}X \para X \para P)
\]
We claim that 
\begin{itemize}
\item[(1)] $!P \;\HOCCSMWB\; !P\para P$; but 
\item[(2)] $!P \;\not\HOCCSMSB\; !P\para P$
\end{itemize}

For part (1) of the claim, it should be clear that every action by $!P$ can be matched by $!P\para P$ because the former is a component of the latter. Moreover, an action by $!P\para P$, say $!P\para P \st{\lambda} T$ where $\lambda$ can be $\tau$ or visible, can be matched by $!P \st{\tau} !P\para P \st{\lambda} T$. 

For part (2) of the claim, suppose that $P$ can make an action $P\st{\lambda'} P'$ other than those over $c$, e.g., $P\DEF \overline{d}0.0$. Then $!P\para P \st{\lambda'} !P\para P'$ cannot be matched by $!P$ because $!P$ can only fire immediate visible actions over $c$. 
\end{proof}

%\noindent\textit{Remark}~ 
%We remark that Lemma \ref{l:hopi_bisi_ws_coin} actually demonstrates that the second inclusion of Lemma \ref{l:hoccsm_bisi_rels} is strict, and such is essentially owing to the replication that can be derived in the higher-order setting of \HOCCSm. Indeed, the ability of \HOCCSm\ to encompass recursive behaviour somewhat compensates the absence of the restriction operator, in effect to attain some involved interactions. 
As demonstrated by Lemma \ref{l:hopi_bisi_ws_coin}, that the second inclusion of Lemma \ref{l:hoccsm_bisi_rels} is strict is essentially attributed to the fact that the replication that can be derived in \HOCCSm. Indeed, the capacity of \HOCCSm\ to encode recursive behaviour somewhat compensates the loss in expressiveness incurred by the absence of the restriction operator. 
%\xxc{For instance, concerning context bisimilarity the following two \HOCCSm\ processes are equal: $\overline{a}[!g].(!g\para !h)$ and $\overline{a}[!h].(!g\para !h)$ (one can safely assume that the replications $!g$ and $!h$ are built using different fresh replicator names).
%}
We remark that actually the first inclusion of Lemma \ref{l:hoccsm_bisi_rels} is strict as well. That is, $\SE \;\subsetneq\; \HOCCSMSB$, where the \emph{difference} between $\SE$ and $\HOCCSMSB$ results from some distributive law \cite{HP08}; see \cite{LPSS10a} for more details. 

\section{On the bisimilation equality in \CCSm}\label{s:ccsm}

In this section, we first define \CCSm, i.e., CCScore without the restriction operation. %\ which stands for CCScore without the restriction operator. It is exactly the standard pure CCS by Milner \cite{Mil89} from which the restriction operator is eliminated.
 Then we discuss the relationship between the strong and weak %bisimulation equalities, i.e., 
 bisimilarities. 
 In particular, we show that the strong bisimilarity is still strictly finer than the weak bisimilarity. However, unlike the situation for \HOCCSm, the distance between the strong bisimilarity and the weak bisimilarity can be shortened. % in a sense. %, in the absence of the restriction operator. 
 This is evidenced by the so-called quasi-strong bisimilarity, which requests more than the weak bisimilarity and moves closer to the strong bisimilarity, %as explained in Section \ref{s:introduction} (), 
 but still turns out to be coincident with the weak bisimilarity.   

\subsection{Calculus \CCSm}

% We denote by \CCSm\ the pure CCS without the restriction operator. 
% The %(standard) 
% syntax is given below. 
The syntax of \CCSm\ is given as follows.
We use capital letters to stand for processes.
\[
P,Q := 0 \;\Big{|}\; a.P \;\Big{|}\; \overline{a}.P \;\Big{|}\; P\para Q \;\Big{|}\; !P
\]

The operational semantics is also standard and presented below (we skip the symmetric rules).
\begin{mathpar}
%({Struc})\quad
\inferrule{ }{a.P \st{a} P} \and 
\inferrule{ }{\overline{a}.P \st{\overline{a}} P} \\
\inferrule{P \st{\overline{a}} P' \and Q \st{a} Q'}{P\para Q \st{\tau} P'\para Q'} \and
\inferrule{P \st{\gamma} P'}{P\para Q \st{\gamma} P'\para Q} \\ %\\
%\inferrule{P\para !P \st{\gamma} P'}{!P \st{\gamma} P'} \and
\inferrule{P\st{\gamma} P'}{!P \st{\gamma} P'\para !P} \and
\inferrule{P\st{a} P' \and P\st{\overline{a}} P''}{!P \st{\tau} P'\para P''\para !P}
\end{mathpar}

There are three kinds of actions (ranged over by $\alpha,\beta,\gamma$): input ($a$), output ($\overline{a}$), and internal ($\tau$). The $\tau$ action is often referred to as silent or invisible, and the others as visible.
We write $\overline{\gamma}$ for the complement of $\gamma$, i.e., $\overline{\gamma}$ is $\overline{a}$ if $\gamma$ is $a$ and $a$ if $\gamma$ is $\overline{a}$. %The same applies to prefixes. 
Sometimes we will omit the trailing $0$ in a process, e.g., $a, \overline{a}$ are shortcuts for $a.0, \overline{a}.0$ respectively. 
Like \HOCCSm, a context $C[\cdot]$ is a process with some subprocess replaced by a hole $[\cdot]$, and $C[Q]$ means substituting the hole in $C$ with process $Q$.
% We say that an action $\alpha$ is state-changing if $P \st{\alpha} P'$ and $P\,\NCCSMWB\, P'$; otherwise it is state-preserving. We distinguish between these two kinds of actions because they are significant for the incoming discussion.
We reuse $\SE$ to stand for the standard structural congruence for \CCSm\ \cite{Mil89,SW01a}, like in \HOCCSm\ except that the rule for $\alpha$-convertibility disappears.
A name is said to be fresh if it does not appear in the current processes. 
Other conventions in \HOCCSm\ are carried over here (e.g., $\n{\cdot}$) and we will use them in need without further notice; this shall not cause confusion in contexts. 
A process is divergent if it can fire an infinite sequence of $\tau$ actions, e.g., $!(\overline{a} \para a)$ and $!\overline{a} \para !a$. 
Said otherwise, if a process is not divergent, then it can only engage finitely many internal actions. %interactions.

\subsection*{\textit{Bisimulation}}

We now define the bisimulations, strong and weak. It should be noted that the bisimulations here explicitly consider divergence property \cite{SW01a,San11,San12}. We say that a relation \R\ is divergence-sensitive, %\bc{(or co-divergent ?)},
if for every $P \,\R\, Q$, $P$ diverges if and only if $Q$ does.
%Divergence was initially not considered in the original bisimulation in CCS, but later found its importance \cite{SW01a,San12}. 
%To allow for more flexibility and make the results more interesting, we impose divergence sensitivity in the bisimulation of this work. Besides, divergence sensitivity appears to be a strong requirement to some extent, since it requires a pair of processes to simultaneously diverge or not unconditionally.
We impose divergence sensitivity in the bisimulation as it makes sense to distinguish between divergent and non-divergent processes, particularly from the standpoint of programming languages. Specifically, divergence sensitivity requires a pair of processes to simultaneously diverge or not unconditionally.

% \sepp\sepp
% \fbox{\nts{\large \xxa{RE-RE-RE-ENGAGE FROM \rc{$\clubsuit\clubsuit$ HERE HERE HERE!!!}: 
% %(0) ADD the missing parts (e.g., Defs); 
% }}}
% \sepp\sepp

\begin{definition}
A symmetric binary relation \R on \CCSm\ processes is a strong (respectively weak) bisimulation if it is divergence-sensitive and whenever $P \R Q$, it holds that
\begin{itemize}
%\item \R is divergence-sensitive \bc{(or co-divergent ?)}.
% \item if $P\st{a} P'$, then $Q\st{a} Q'$ (respectively $Q\wt{a} Q'$) and $P'\R Q'$;
% \item if $P\st{\overline{a}} P'$, then $Q\st{\overline{a}} Q'$ (respectively $Q\wt{\overline{a}} Q'$) and $P'\R Q'$;
% \item if $P\st{\tau} P'$, then $Q\st{\tau} Q'$ (respectively $Q\wt{} Q'$) and $P'\R Q'$.
\item if $P\st{\alpha} P'$, then $Q\st{\alpha} Q'$ (respectively $Q\wt{\widehat{\alpha}} Q'$) and $P'\R Q'$;
\end{itemize}
Two processes $P$ and $Q$ are strongly (respectively weakly) bisimilar, notation $P \,\CCSMSB\, Q$ (respectively $P \,\CCSMWB\, Q$), if there exists some strong (respectively weak) bisimulation \R such that $P \,\R\, Q$.
\end{definition}
\sepp

We call $\CCSMWB$ and $\CCSMSB$ the weak bisimilarity and strong bisimilarity respectively. As is well-known, they are equivalence relations and congruences.
%adept at preserving the operations of the calculus, i.e., both of them are congruences, particularly as opposed to the involved situation in the higher-order paradigm (viz., relatively hard to prove congruence properties). 
In the discussion of the bisimilarities, we may use the up-to techniques to build bisimulations, e.g., bisimulation up-to context.  These are well-established proof method for process models; see \cite{SW01a} for a comprehensive introduction. 
For the sake of convenience, we give the definition of the (weak) bisimulation up-to context. We shall see how it is used in the passing. % upon using.
\begin{definition}\label{def:up-to-context-ccs}
A symmetric binary relation \R on \CCSm\ processes is a (weak) bisimulation up-to context if it is divergence-sensitive and whenever $P \R Q$, it holds that
\begin{itemize}
\item if $P\st{\alpha} P'$, then $Q\wt{\widehat{\alpha}} Q'$ and there exist some context $C$, $P_1$ and $Q_1$ such that 
\[
P' \SE C[P_1], \quad Q' \SE C[Q_1], \quad \mbox{ and }\quad P_1 \,\R\, Q_1
\]
\end{itemize}
\end{definition}
As the following lemma states, if a relation is a bisimulation up-to context, then it is subsumed by the weak bisimilarity. The proof of this lemma amounts to showing that the weak bisimilarity is contextual (i.e., preserved by contexts), and divergence sensitivity would not raise obstacle because it does not involve explicit action matching; see \cite{SW01a} for a reference of proof and more discussion.  
\begin{lemma}\label{l:bisi-upto-context}
If \R\ is a bisimulation up-to context, then it holds that $\R\ \subseteq\ \CCSMWB$.
\end{lemma}
\sepp

In terms of the bisimilarity, actions can be divided into two classes: state-changing and state-preserving. We say that an action $\alpha$ as occurring in $P \st{\alpha} P'$ is state-changing if $P\,\NCCSMWB\, P'$; otherwise it is state-preserving. We distinguish between these two kinds of actions because they are significant for the incoming arguments.

It should be clear that the following implications are true \cite{Mil89,SW01a}. %We note that $\SE$ stands for the (standard) structural congruence, like in \HOCCSm.  
\begin{lemma}\label{l:ccsm_bisi_rels}
It holds that $\SE \;\subsetneq\; \CCSMSB \;\subseteq\; \CCSMWB$.
\end{lemma}
Both of the implications of Lemma \ref{l:ccsm_bisi_rels} are immediate from the definitions. 
%For the first implication of that lemma to be strict, we notice that $!a(x).0$ and $!a(x).0 \para !a(x).0$ are strong bisimilar but not structural congruent.
For the first implication of the lemma to be strict, we notice that $!a.0$ and $!a.0 \para !a.0$ are strongly bisimilar but not structurally congruent.

%% !TEX root = ./main.tex

%remark about the relationship between the structural congruence and the strong context bisimilarity

\iftoggle{hidingON}{%
  % use hiding

} {%
  % no hiding
This is a good moment to contrast Lemma \ref{l:ccsm_bisi_rels} with Lemma \ref{l:hoccsm_bisi_rels}.
For Lemma \ref{l:hoccsm_bisi_rels}, the first inclusion of it is not clear to be strict. 
%For the first implication of that lemma, we do not know if it is strict too. 
That being said, we conjecture that $\HOCCSMSB$ coincides with $\SE$.
% \begin{conjecture}\label{conj:hoccs_struc_striin_strong}
% It holds that $\SE \,=\, \HOCCSMSB$.
% \end{conjecture}
%\noindent 
%Jumping ahead a little bit, 
A guiding intuition is that  
the accompanying examples used for Lemma \ref{l:ccsm_bisi_rels} above do not apply to Lemma \ref{l:hoccsm_bisi_rels}. The major reason is that the derived replication in \HOCCSm\ brings some extra moves in order to duplicate a process. 
To be more specific, in the case for \CCSm, %CCS without the restriction operator, 
we observe that $!a.0$ and $!a.0 \para !a.0$ are strong bisimilar but not structural congruent. These two processes can be defined in \HOCCSm, as below. 
\[
\begin{array}{lcl}
P_1 &\DEF& !a.0 \\
 &\SE& \overline{c}Q_{\scriptscriptstyle c,a.0} \para Q_{\scriptscriptstyle c,a.0}, \\ %\qquad\qquad\mbox{ in which }\\
 Q_{\scriptscriptstyle c,a.0} &\DEF& c(X).(\overline{c}X \para X \para a.0) \\\\
P_2 &\DEF& P_1 \para P_1
\end{array}
\]
However, $P_1$ and $P_2$ are not strong context bisimilar, let alone structural congruent. %This becomes clearer when we expand their definition, as below.
One can see that $P_2$ can have two consecutive output on $c$ while $P_1$ cannot match.
%though we we notice that $!a(x).0$ and $!a(x).0 \para !a(x).0$ are strong bisimilar but not structural congruent.
Following this, processes like $\overline{a}[!g].(!g\para !h)$ and $\overline{a}[!h].(!g\para !h)$ are not context bisimilar either (one can safely assume that the replications $!g$ and $!h$ are built using different fresh replicator names).

{More concerning the replication and in light of the example above, the derived replication in \HOCCSm\ does not appear to encode the standard primitive replication operator in a reasonably good way. For that purpose, there seems little hope to achieve a perfect interpretation of the replication in a restriction-free setting. 
}

}%toggle hiding

% (NO TRUE!) The first implication of that lemma is strict too. To see this, consider two processes defined as follows.
% \[
% P_1\DEF !a(X).0 \qquad\qquad P_2\DEF !a(X).0 \para !a(X).0
% \]
% It is strightforward that $P_1$ and $P_2$ are strong bisimilar, because every action ($\tau$ or input on $a$) by $P_2$ can be matched by $P_1$. Yet they are not structural congruent, because no rules of the structural congruence can equate them.

%We look into the conjecture below. 
%negative direction
%\input{appendix_conjecture_negative_direction.tex}
%positive direction
%\input{appendix_conjecture_positive_direction.tex}

%\subsection{Extension: \rc{Does Conjecture \ref{app:conjecture_hopim} hold for first-order restriction-free calculi, like CCS / pi without resutrction?}}

%\nts{TODO: in another job?! MOVE these (related) stuff to the new git repo ``bisi\_w2s" .}

%
%\clearpage
%

%---------------------------
% Local Variables:
% mode: LaTeX
% TeX-master: "main.tex"
% End:

\sepp

Now we define a bisimulation that stands in between the strong and weak bisimilarities. As will be shown, it coincides with the weak bisimilarity and is slightly weaker than the strong bisimilarity (thus so is the weak bisimilarity). 
\begin{definition}\label{d:quasi-strong-bisi}
%TOADAPT(adapt from the definition of weak bisimulation)
A symmetric binary relation \R on \CCSm\ processes is a quasi-strong bisimulation if it  is divergence-sensitive, and whenever $P \R Q$, the following properties hold. 
\begin{itemize}
\item if $P\st{\alpha} P'$ and $\alpha$ is not $\tau$, then $Q\wt{}\st{\alpha} Q'$ and $P'\R Q'$.
\item if $P\st{\tau} P'$, then $Q\st{\tau} Q'$ and $P'\R Q'$.
\end{itemize}
Two processes $P$ and $Q$ are quasi-strongly  bisimilar, notation $P \,\CCSMQSB\, Q$, if there exists some quasi-strong bisimulation \R such that $P\,\R\, Q$.
\end{definition}
For a quasi-strong bisimulation, since internal actions are bisimulated in a strong manner, divergence sensitivity somehow becomes a derived condition, and we include it in the definition for convenience.
By definition, it is straightforward to see that $\CCSMSB \;\subseteq\; \CCSMQSB \,\subseteq\, \CCSMWB$.

\subsection{The relationship between the weak and strong bisimilarities 
%\rc{NEED REWRITING!!!!!!}
in \CCSm}\label{s:rel_w_s_ccs}

We first remark that the approach for Lemma \ref{l:hopi_bisi_ws_coin} does not apply to the case of \CCSm, because the semantics ascertain that $!P\;\CCSMSB\; !P\para P$ (actually this holds in any CCS-like calculus, unless the replication is defined differently). This observation lends confidence to the coincidence between $\CCSMWB$ and $\CCSMSB$. {But this turns out to be wrong.} We will elaborate this in the current section.

% \sepp\sepp
% \fbox{\nts{\large \xxa{RE-RE-RE-ENGAGE FROM \rc{$\clubsuit\clubsuit$ HERE HERE HERE!!!}: 
% %(0) ADD the missing parts (e.g., Defs); 
% }}}
% \sepp\sepp

In \CCSm, the weak bisimilarity is still strictly coarser than the strong bisimilarity. The intrinsic reason is that we still have the replication operator, from which infinite behaviour arises. This renders false the  matching of visible actions in a strong manner (the best we can have is as Proposition \ref{p:tau_case} illustrates).

To see a counterexample, take 
\[
\begin{array}{lcl}
P_1 &\DEF& !c.d \para !\overline{c} \para d \\
P_2 &\DEF& !c.d \para !\overline{c} \para !c
\end{array}
\]
Obviously $P_1 \,\NCCSMSB\, P_2$. However $P_1 \,\CCSMWB\, P_2$, because the action 
\[P_1\;\st{d}\; !c.d \para !\overline{c} \para 0
\] can be simulated by
% \[P_2\;\st{\tau}\; !c.d \para !\overline{c} d\para 0 \para!c \;\st{d}\; !c.d \para !\overline{c} 0\para 0 \para!c
% \] 
\[P_2\;\st{\tau}\; !c.d \para d\para 0 \para !\overline{c}\para !c \;\st{d}\; !c.d \para 0 \para 0 \para !\overline{c} \para !c
\]
Henceforth, every $d$ produced by $P_1$ in simulating the subprocess $!c$ in $P_2$ can be simulated in a similar way.
So we have the following proper inclusion.
\begin{lemma}\label{con:ccs_bisi_ws_coin}
In \CCSm, $\CCSMSB \,\subsetneq\, \CCSMWB$. 
\end{lemma}

In the remainder of this section, we prove the follow-up theorem (Theorem \ref{con:ccsm_bisi_ws_coin}), which says that the weak bisimilarity can somehow be %(perfectly) 
approximated by a partially strong bisimilarity, i.e., the quasi-strong bisimilarity. It reveals that in \CCSm, there is truly some stronger characterization of the weak bisimilarity, or in other words, the gap between the strong and weak bisimilarities is diminished to some (arguably) noticeable extent. In a broader sense, this exhibits that the weak bisimilarity can indeed be strengthened in a process model without the restriction operator, at least in the first-order paradigm.

\begin{theorem}\label{con:ccsm_bisi_ws_coin}
Assume $P$ and $Q$ are \CCSm processes. Then $P \,\CCSMWB\, Q$ implies $P \,\CCSMQSB\, Q$. 
\end{theorem}
This theorem immediately leads to the corollary below.
\begin{corollary}\label{cor:coin-weak-quasi-strong}
In \CCSm, it holds that $\CCSMWB \;=\; \CCSMQSB$.
\end{corollary}
\sepp

To establish the implication claimed by Theorem \ref{con:ccsm_bisi_ws_coin}, we 
%make some preparation %for  the  (and hence the coincidence). These preparation 
%that 
exploit %deeply 
the structure of \CCSm\ processes, and %meanwhile 
go through multiple analyses %(both syntactic and semantic) 
toward our goal, i.e., the weak bisimilarity is indeed a quasi-strong bisimulation.
\sepp

% \fbox{\tiny \rc{TO REMOVE!}
% \begin{minipage}{10cm}
% \nts{TODO (OLD): Maybe follow the line in Appendix \ref{app:conj_pos_ref}, which is a failure trial for \HOPim, but may be applicable to CCS. TO SEE ...}\\\\
% \nts{\rc{Also} notice the possibility of proving $\CCSMWB \;\subseteq\; \SE$, referring to \cite{Fu15} for the proof of the coincidence between bisimilarity and structural congruence in $\mathbb{C}$ \\
%  (if so, we need to work on some \textbf{`normal form'} of a \CCSm\ process (what does it look like? composition of input/output and replicated input/output?), to show two weakly bisimilar processes share the same normal form). {\small \rc{but} this seems hard in presence of replication}}

% \sepp
% \nts{ BELOW START a trial path ... (noticing the above framed remarks) ...}
% \sepp
% \end{minipage}
% }\sepp

%\sepp

%-----------------------------------------------------------------
%\sepp

% \sepp\sepp
% \fbox{\nts{\large \xxa{CLEAR FROM \rc{$\clubsuit\clubsuit$ HERE HERE HERE!!!}: 
% %(0) ADD the missing parts (e.g., Defs); 
% }}}
% \sepp\sepp

%\SEPALine

% \sepp\sepp
% \fbox{\nts{\large \xxa{RE-RE-RE-ENGAGE FROM \rc{$\clubsuit\clubsuit$ HERE HERE HERE!!!}: 
% %(0) ADD the missing parts (e.g., Defs); 
% }}}
% \sepp\sepp

\iftoggle{hidingON}{%
  % use hiding
} {%
  % no hiding
%As a warm-up, 
As a warm-up and the first observation, we capture a structural property in Lemma \ref{l:ccs_tdiv_struc}. Although this lemma is not directly used afterwards, it provides good intuition and analytical thought helping us to go through the discussion in what follows.

\begin{lemma}\label{l:ccs_tdiv_struc}
\nts{IS this lemma necessary? NOT used elsewhere! \rc{MAybe to remove...TO DECIDE soon!...}}\\
%\item[(1)] 
Assume $P$ is not divergent and makes an infinite number of some visible action $\alpha$ (possibly separated by a number of $\tau$ actions, but no other visible actions), that is, $P\wt{}\st{\alpha}\wt{}\st{\alpha}\wt{} \cdots \wt{}\st{\alpha} \cdots$.
Then $P$ must have a sub-process of the following form (up-to $\CCSMSB$):
\[
!\gamma_1.\gamma_2....\gamma_k.(\alpha.P'\para P'') \para P''' \quad (k\geqslant 0)
\] where
%\sepp\sepp

% \begin{enumerate}
% %possibility of the general form of the subprocess up-to $\CCSMSB$:
% \item %(1)
% $!\gamma_1.\gamma_2....\gamma_k.(\alpha.P'\para P'') \para P'''$ ($k\geqslant 0$), in which $P'''$ and at least one of $P'$ and $P''$ can make the action sequence $\wt{\gamma_1}\wt{\gamma_2}\cdots \wt{\gamma_k}$.
% This possibility may occur in a parallel composition. \nts{So check the case for parallel composition, tanks this into consideration...}
% \item %(2) 
% $!(\gamma_1.\gamma_2....\gamma_k.(\alpha.P'\para P'') \para P''')$ ($k\geqslant 0$), in which $P'''$ can make the action sequence $\wt{\gamma_1}\wt{\gamma_2}\cdots \wt{\gamma_k}$.
% This possibility may occur in a replication. 
% \end{enumerate}
\begin{enumerate}
%possibility of the general form of the subprocess up-to $\CCSMSB$:
\item %(1)
$P'''$ can make the action sequence 
%\cancel{$\wt{\gamma_1}\wt{\gamma_2}\cdots \wt{\gamma_k}$} ~
\ntsvv{$\wt{\overline{\gamma_1}}\wt{\overline{\gamma_2}}\cdots \wt{\overline{\gamma_k}}$}, and moreover, 

% \nts{Should it be the following? }

% $P'''$ can make the action sequence $\wt{\overline{\gamma_1}}\wt{\overline{\gamma_2}}\cdots \wt{\overline{\gamma_k}}$, and moreover,

\item if $P'''$ cannot have infinitely many occurrences of that sequence, then at least one of $P'$ and $P''$ can make the action sequence 
%\cancel{$\wt{\gamma_1}\wt{\gamma_2}\cdots \wt{\gamma_k}$} 
\ntsvv{$\wt{\overline{\gamma_1}}\wt{\overline{\gamma_2}}\cdots \wt{\overline{\gamma_k}}$}.

% \nts{Should it be the following? }

% if $P'''$ cannot have infinitely many occurrences of that sequence, then at least one of $P'$ and $P''$ can make the action sequence $\wt{\overline{\gamma_1}}\wt{\overline{\gamma_2}}\cdots \wt{\overline{\gamma_k}}$.
\end{enumerate}
% \sepp\sepp\sepp
% %$!\gamma_1.\cdots\gamma_k.\alpha.P'\para P''$ ($k\geqslant 0$), \\
% %$!(\gamma_1.\cdots\gamma_k.\alpha.P'\para P'')$ ($k\geqslant 0$), \\
% %$!\gamma_1.\cdots\gamma_k.\alpha.(P'\para P'') \para P''$ ($k\geqslant 0$), \\
% $!\gamma_1.\cdots\gamma_k.(\alpha.P'\para P'') \para P''$ ($k\geqslant 0$), \\
% %or %\nts{?...} $\alpha.!\gamma_1.\cdots\gamma_k.\alpha.P'\para P''$ ($k\geqslant 0$), \\
% or $!\alpha.P'$ (actually a special case of the above? yes, when $k{=}0$ and $P''\SE 0$) \nts{?...\rc{rewrite??...}} \\
% %up-to $\CCSMSB$ (or $\SE$?), 
% up-to $\CCSMWB$ \nts{(or is $\CCSMSB$ or $\SE$ possible? seems no.)} , 
% %for some $P'$ and $P''$ such that both of them are capable of firing actions $\wt{\overline{\gamma_1}}\cdots\wt{\overline{\gamma_k}}$.
% for some $P'$, and $P''$ that is capable of firing actions $\wt{\overline{\gamma_1}}\cdots\wt{\overline{\gamma_k}}$. \fbox{\nts{NEED more CHECK!! } \rc{HERE!}} 
\end{lemma}
\begin{proof} %\fbox{\nts{NEED REDO and more CHECK!! } \rc{HERE!}}
Assume 
\begin{equation}\label{eq:infi_vis_action}
P\wt{}\st{\alpha}\wt{}\st{\alpha}\wt{} \cdots \wt{}\st{\alpha} \cdots
\end{equation}
Intuitively, $P$ must be of the form $C[!C_1[\alpha.P_1]]$ for some contexts $C$, $C_1$, and $P_1$, where $\alpha$ can be fired right after some interactions within $C$. 
Hereby we  proceed using %transition 
induction on $P$. % \nts{TODO??}.
\begin{itemize}
\item $P$ is $0$. This is impossible for the premise to hold.

\item $P$ is $a.P_1$. %\nts{(Is this enough?? seems yes)}
In this case, $\alpha$ is $a$ and $P\st{a} P_1\wt{}\st{a}\cdots$. By induction hypothesis (ind. hyp. for short henceforth), $P_1$ must have a subprocess of the designated forms. So does $P$.

\item $P$ is $\overline{a}.P_1$. %\nts{(Is this enough?? seems yes)}
In this case, $\alpha$ is $\overline{a}$ and $P\st{\overline{a}} P_1\wt{}\st{\overline{a}}\cdots$. By ind. hyp., $P_1$ must have a subprocess of the designated forms. So does $P$.

\item $P$ is $P_1\para P_2$. %\nts{(Is this enough?? seems yes)}
At least one of $P_1$ and $P_2$ must have an infinite number of some visible action $\alpha$ (possibly separated by a number of $\tau$ actions, as described in the premise of this lemma). By ind. hyp., it must have subprocesses of the designated forms. So does $P$.

There are two cases.
\begin{enumerate}
\item If $P_1$ or $P_2$ alone has the infinite $\alpha$ actions as incurred, then we are done by ind. hyp..

\item %\fbox{\nts{NEED more CHECK!! \rc{HERE!}} }
Otherwise, if that infinite $\alpha$ actions involve some interactions between $P_1$ and $P_2$, i.e., the $\tau$ actions in the sequence (\ref{eq:infi_vis_action}) result from the interactions between $P_1$ and $P_2$. In this case, $P_1$ must be of the form $!\gamma_1.\gamma_2....\gamma_k.(\alpha.P'\para P'')$  ($k\geqslant 0$), and $P_2$ must be able to perform the actions 
%\cancel{$\wt{\gamma_1}\wt{\gamma_2}\cdots \wt{\gamma_k}$}~ 
\ntsvv{$\wt{\overline{\gamma_1}}\wt{\overline{\gamma_2}}\cdots \wt{\overline{\gamma_k}}$}.  
This in turn has two subcases.
\begin{enumerate}
\item If $P_2$ is able to perform 
%\cancel{$\wt{\gamma_1}\wt{\gamma_2}\cdots \wt{\gamma_k}$}~ 
\ntsvv{$\wt{\overline{\gamma_1}}\wt{\overline{\gamma_2}}\cdots \wt{\overline{\gamma_k}}$} infinitely many times, then this ensures the sequence (\ref{eq:infi_vis_action}).
\item Otherwise if $P_2$ is unable to perform \ntsvv{$\wt{\overline{\gamma_1}}\wt{\overline{\gamma_2}}\cdots \wt{\overline{\gamma_k}}$} infinitely many times, then to have the sequence in (\ref{eq:infi_vis_action}), at least one of $P'$ and $P''$ must be capable of doing 
%\cancel{$\wt{\gamma_1}\wt{\gamma_2}\cdots \wt{\gamma_k}$} ~ 
\ntsvv{$\wt{\overline{\gamma_1}}\wt{\overline{\gamma_2}}\cdots \wt{\overline{\gamma_k}}$} (not necessarily infinitely many times though) .
\end{enumerate}
So in summary, $P_1\para P_2$ indeed has the desired form.
\end{enumerate}

\item $P$ is $!P_1$. There are two cases. 
\begin{itemize}
\item If $P_1$ itself has an infinite number of some visible action $\alpha$ (possibly separated by a number of $\tau$ actions). Then by ind. hyp., we are done.

\item %\fbox{\nts{NEED more CHECK!! \rc{HERE!}} }
Otherwise, if $P_1$ does not have any infinite number of some visible action $\alpha$,  % (possibly separated by a number of $\tau$ actions). 
then the infinite $\alpha$ actions must come from the synergy of $P_1$ and the outmost replication operation. 
In general, $P_1$ must take the form $\gamma_1.\cdots\gamma_k.(\alpha.P'\para P'')\para P''$ ($k\geqslant 0$). % up-to $\CCSMSB$ \nts{(or $\SE$?Seems no!)}. 
It satisfies the wanted form, because $!(\gamma_1.\cdots\gamma_k.(\alpha.P'\para P'')\para P'') \,\CCSMSB\, !\gamma_1.\cdots\gamma_k.(\alpha.P'\para P'') \para !P''$.
Notice that when $k$ is zero, it is of the (special) form $!\alpha.P'\para P''$.
% meeting the requirement. When $k$ is non-zero,  
%\nts{Is this ok and enough?? maybe yes nw.}
\end{itemize}

\end{itemize}
\end{proof}
}%hiding toggle

%\sepp

%-----------------------------------------------------------------
%\SEPALine
\sepp

The following lemma claims that every $\tau$ action changes the interactional capability, i.e., the `state' of a process,  except for those with infinite actions. We recall that a process $P$ diverges if it has an infinite sequence of $\tau$ actions, i.e., $P\st{\tau}\cdots\st{\tau}\cdots$, and we say that a process $P$ has an infinite number of visible action $\alpha$ ($\alpha$ is not $\tau$) if $P\wt{}\st{\alpha}\wt{}\st{\alpha}\cdots \wt{}\st{\alpha}\cdots$, among which a noticeable special case is $P\st{\alpha}\st{\alpha}\cdots \st{\alpha}\cdots$. % when $P$ must take the form $\alpha.P'\para P''$ (by a simple induction on $P$). 

% \sepp\sepp
% \fbox{\nts{\large \xxa{RE-RE-RE-ENGAGE FROM \rc{$\clubsuit\clubsuit$ HERE HERE HERE!!!}: 
% %(0) ADD the missing parts (e.g., Defs); 
% }}}
% \sepp\sepp

\begin{lemma}\label{l:ccs_tau_state}
%If $P$ is not capable of an infite number of action (visible or silent) and $P\st{\tau} P'$, then $P \,\NCCSMWB\, P'$.
If $P$ is not divergent and $P\st{\tau} P'$, then $P \,\NCCSMWB\, P'$.
\end{lemma}
\begin{proof}
%\fbox{\nts{NEED more CHECK!!} } 

First, we note that the only possibility of having infinite actions is through replication, since the replication operator is the mere way of producing infinity, by a simple induction. 
We have the following observations that lead to the result of this lemma.  %\fbox{\nts{NEED more CHECK!! \rc{HERE!}} }
%\nts{Some main points:}
\begin{enumerate}
\item \label{l:ccs_tau_state_1}
If $P$ is not divergent, then for any $a$, $P$ cannot have both an infinite number of $\overline{a}$ actions and an infinite number of $a$ actions (notice that only $\tau$ actions occur between each two neighbouring visible actions). 
Assume for a contradiction that $P$ has both an infinite number of $\overline{a}$ actions and an infinite number of $a$ actions. Then there are two possibilities.
\begin{enumerate} 
\item \label{l:ccs_tau_state_1a}  
These two threads, i.e., infinite numbers of $\overline{a}$ actions and $a$ actions respectively, are in parallel composition. That is, $P\SE P_1\para P_2$ in which $P_1$ is capable of an infinite number of $\overline{a}$ actions and $P_2$ is capable of an infinite number of $a$ actions; 
\item \label{l:ccs_tau_state_1b}  
The two thread are intertwined, i.e., $P\wt{}\st{\gamma_1}\wt{}\st{\gamma_2}\cdots \wt{}\st{\gamma_j}\wt{}\cdots$ in which $\gamma_i$ is either $\overline{a}$ or $a$. 
\end{enumerate}
We claim that both case (a) and case (b) would lead to the divergence of $P$, a contradiction. Case (a) is obvious. For case (b), to yield the infinite visible action sequence (separated by $\tau$ actions only), all the actions $\gamma_i$ must be consumable through interactions, i.e., each $\gamma_i$ should go away by taking part in some interaction. This immediately results in divergence. %Moreover thanks to (1), 

\item \label{l:ccs_tau_state_2}
Due to \ref{l:ccs_tau_state_1}, if $P$ is not divergent and can make a $\tau$ that comes from an interaction, say over $a$, between its parallel components, then $P\st{\overline{a}}$ and also $P\st{a}$ but not both of these two actions incur infinity of the same action. This, in turn, means that $P$ either has a finite (non-zero) number of $\overline{a}$ actions or a finite number of $a$ actions. Suppose it is the latter, and the former is similar. Then we conclude that $P$ cannot be bisimilar to $P'$ because $P'$ is short of one action $a$ for such a name $a$. This suffices to obtain the result because once that finite visible action is fired it never comes back.    
\end{enumerate} \vspace*{-.5cm}
\end{proof}

We have a straightforward corollary of Lemma \ref{l:ccs_tau_state}, since $\CCSMSB$ is included in $\CCSMWB$.
%since $\CCSMWB$ is no finer than $\CCSMSB$.
\begin{corollary}\label{cor:state_change}
If $P$ is not divergent and $P\st{\tau} P'$, then $P \,\NCCSMSB\, P'$.
\end{corollary}
\sepp
Another corollary also follows straight away.
\begin{corollary}\label{cor:state_change2}
If $P\st{\tau} P'$ and $P \,\CCSMSB\, P'$ (or $P \,\CCSMWB\, P'$), then $P$ is divergent.
\end{corollary}
\sepp

The following lemma describes that $\tau$ actions resulting from a replication in a process do not change the state of that process. 

\begin{lemma}\label{l:rep_invar}
%For any $P'$ s.t. $!P\wt{} \st{\alpha_1} \wt{} \st{\alpha_2} \wt{} \cdots \wt{} \st{\alpha_k}\wt{} P'$, it holds that $!P \,\CCSMWB\, P'$. 
%In particular, 
Whenever $!P \wt{} P'$, it holds that $!P \,\CCSMWB\, P'$.
\end{lemma}
\begin{proof}%[Proof of Claim \ref{l:rep_invar}]
%Suppose $!P\wt{} P'$, we show $!P \,\CCSMWB\, P'$ by induction on the length of $\wt{}$, i.e., the number of $\tau$ actions in it. Obviously, the result holds when the length is zero. %Now suppose the result holds for length $n$, we need to show it also holds for $n+1$. To see this, \nts{TODO...}

%(2) use induction on the length of $\wt{}$, to show that $!P \CCSMWB P'$ for every $P'$ s.t. $!P \wt{} P'$.  (\nts{MAKE this property into a lemma?})

We first look at a particular case, and the general one can be proven in essentially %more or less 
the same way.

%\nts{TODO: elaborate and/or rewrite (the bisimulation building part) ...}

\begin{enumerate}
\item %1) 
In particular, we first show the result for the case when the length of the $\tau$ action by $!P$ is $1$, i.e., $!P\st{\tau} P'$. \\
From the semantics, the $\tau$ action by $!P$ must be one of the following cases. %scenarios.
% \[
% \begin{array}{lll}
% !P \st{\tau} P_1\para !P \;\DEF P_1' &\qquad \mbox{ due to }&\quad P \st{\tau} P_1 \qquad \mbox{ or } \\
% !P \st{\tau} P_1\para P_2\para !P \;\DEF P_1'' &\qquad \mbox{ due to }&\quad P \st{a} P_1 \mbox{ and } P\st{\overline{a}} P_2
% \end{array}
% \] 

\begin{enumerate}
\item $!P \st{\tau} P_1\para !P$ and $P_1' \DEF P_1\para !P$ due to $P \st{\tau} P_1$. \\
We show that the following relation \R is a weak bisimulation up-to context (Definition \ref{def:up-to-context-ccs}; see also \cite{SW01a}), so that $!P \CCSMWB P_1'$. 
\[
\R \DEF \{(!P, P_1')\} \,\cup\, \CCSMWB
\]
Assume $!P \,\R\, P_1'$. We have two simulation scenarios.
\begin{enumerate} 
\item %Now we have $!P \CCSMWB P_1'$ because everything $P_1'$ can do can be simulated by $!P$ by first making the $\tau$ action and do whatever $P_1'$ does. 
Suppose $P_1'\st{\alpha} P_1''$. Then $!P$ simulates by $!P \st{\tau} P_1' \st{\alpha} P_1''$, and $P_1'' \,\R\, P_1''$ because $P_1'' \,\CCSMWB\, P_1''$. In the spirit of ``up-to context", one simply sets the context to be $[\cdot]$.

\item Conversely, suppose $!P \st{\alpha} P_2\para !P\DEF P_2'$ (when $\alpha$ is $\tau$, we may assume $P_1$ is not $P_2$ because otherwise the simulation is trivial). %  due to $P \st{\alpha} P_2$ in which $P_1$ is not $P_2$. 
Then $P_1'$ simulates by 
\[\begin{array}{lcl}
P_1' &\SE& P_1\para !P \st{\alpha} P_1\para P_2\para !P \\
 &\SE& P_2\para P_1\para !P \\
 &\SE& P_2 \para P_1' 
\end{array}
\] The simulation now continues by taking advantage of the ``up-to context". By setting $C\DEF P_2\para [\cdot]$, we have
\[
\begin{array}{l}
P_2' \SE P_2\para !P \SE C[!P] \\\\
P_1\para P_2\para !P \SE P_2 \para P_1' \SE C[P_1']
\end{array}
\] in which $!P \,\R\, P_1'$. So we are done.
\end{enumerate}

\item $!P \st{\tau} P_1\para P_2\para !P \;\DEF P_1''$ due to $P \st{a} P_1$ and $P\st{\overline{a}} P_2$. \\
We show that the following relation $\R'$ is a weak bisimulation up-to context, which implies that $!P \CCSMWB P_1''$.
\[
\R' \DEF \{(!P, P_1'')\} \,\cup\, \CCSMWB
\]
Assume $!P \,\R'\, P_1''$. We have two simulation scenarios.
\begin{enumerate} 
\item %Now we have $!P \CCSMWB P_1'$ because everything $P_1'$ can do can be simulated by $!P$ by first making the $\tau$ action and do whatever $P_1'$ does. 
Suppose $P_1''\st{\alpha} P_1'''$. Then $!P$ simulates by $!P \st{\tau} P_1'' \st{\alpha} P_1'''$, and $P_1''' \,\R'\, P_1'''$ follows due to $P_1''' \,\CCSMWB\, P_1'''$. One can set the context to be $[\cdot]$ to comply with the requirement of ``up-to context".

\item Conversely, suppose $!P \st{\alpha} P_3\para !P\DEF P_3'$ (we may assume $P_3$ is not $P_1\para P_2$ when $\alpha$ is $\tau$ because otherwise the simulation is trivial). %  due to $P \st{\alpha} P_2$ in which $P_1$ is not $P_2$. 
Then $P_1''$ simulates by 
\[
\begin{array}{lcl}
P_1'' &\SE& P_1\para P_2\para !P \st{\alpha} P_1\para P_2\para P_3\para !P \\
 &\SE& P_3\para P_1\para P_2\para !P \\
 &\SE& P_3 \para P_1'' 
\end{array}
\] Taking advantage of the ``up-to context" and setting $C'\DEF P_3\para [\cdot]$, we have
\[
\begin{array}{l}
P_3' \SE P_3\para !P \SE C'[!P] \\\\
P_1\para P_2\para P_3\para !P \SE P_3 \para P_1'' \SE C'[P_1'']
\end{array}
\] in which $!P \,\R\, P_1''$. So we are done.
\end{enumerate}

\end{enumerate}

\item %2) 
% Suppose the result is true for $1$ through $n$. That is, for $!P \st{\tau}_k P'$ ($k\leqslant n$), it holds that $!P \,\CCSMWB\, P'$. 
% We now show the result for $n+1$. Suppose $!P \st{\tau}_{n+1} P'$. It can be rewritten as
% \[
% !P \;\st{\tau}_{n}\; P'' \;\st{\tau}\; P'  \qquad \mbox{ for some } P''
% \]
% By induction hypothesis, $!P\,\CCSMWB\, P''$. Then the action $P'' \;\st{\tau}\; P'$ can be simulated 

Now we show the lemma in its general case, i.e., $!P \wt{} P'$. We remark that by a routine (transition) induction, it can be shown that $P'$ must be of the form $P'' \para !P$ for some $P''$ (up-to $\SE$). We can obtain $!P \,\CCSMWB\, P'$ by showing the following relation $\R_1$ to be a weak bisimulation up-to context.
\[
\R_1 \DEF \{(!P, P')\} \,\cup\, \CCSMWB
\]
To achieve this and complete the proof, we make the arguments in the way exactly as case 1 above in showing $\R$ or $\R'$ to be a weak bisimulation up-to context. { The only difference is that there may be more than one $\tau$ in the weak transitions by $!P$, but the form of the derived process $P'$ does not change at all and the arguments stay put. %as we just remarked.
}  
%Hence the proof of the lemma is now completed.
%\nts{TODO: expand... (see if can go through)}
\end{enumerate}\vspace*{-.5cm}
\end{proof}

%ABOUT UP-TO CONTEXT??

\noindent\textit{Remark}~~ It is noteworthy that Lemma \ref{l:rep_invar} is also true in the standard CCS (with restriction), as the proof does not rely on the absence of the restriction operation. However, if $!P$ has some visible actions, then its state is not preserved any longer. That is, for $P'$ such that $!P\wt{} \st{\alpha_1} \wt{} \st{\alpha_2} \wt{} \cdots \wt{} \st{\alpha_k}\wt{} P'$ in which $\alpha_i$ ($i{=}1,...,k$) is visible (i.e., not $\tau$), it does not necessarily hold that $!P \,\CCSMWB\, P'$. The reason is that the action by $!P$ %$!P\st{\alpha_i} P'$
may reveal some action $!P$ cannot match without first doing another visible action. It is not hard to contrive a counterexample, say $P'$ can do some action that $!P$ cannot do. For instance, let $P\DEF !a.c$. Then $!P \st{a} c \para !P \equiv P'$, and $P'\st{c}$ whereas $!P$ cannot.
%Amend Claim11:
%(1) remove the case for visible action, because it does not necessarily hold that $!P\st{\alpha} P'$ in which $\alpha$ is visible and $!P\CCSMWB P'$ (it is easy to contribution be a counterexample, say $P'$ can do some action $!P$ cannot). But this is true when $!P$ makes $\tau$. 

%\sepp

%-----------------------------------------------------------------
%\SEPALine
\sepp

We continue to examine the state-preserving/state-changing $\tau$ actions in the coming up lemmas.
Again by state-changing, we mean that if $P \st{\alpha} P'$, then $P\,\NCCSMWB\, P'$ (more often than not, $\alpha$ here is $\tau$, though it appears to still make sense if it is visible); otherwise it is state-preserving.

% \sepp\sepp
% \fbox{\nts{\large \xxa{RE-RE-RE-ENGAGE FROM \rc{$\clubsuit\clubsuit$ HERE HERE HERE!!!}: 
% %(0) ADD the missing parts (e.g., Defs); 
% }}}
% \sepp\sepp

\begin{lemma}\label{l:ccs_tau_state_div0}
%If $P$ is not capable of an infite number of action (visible or silent) and $P\st{\tau} P'$, then $P \,\NCCSMWB\, P'$.
Assume $P$ is a \CCSm\ process. Then 
%\begin{enumerate}
%\item 
there exists $k\geqslant 0$ and $P'$ such that $P\st{\tau}_k\; P'$, and $P'\,\CCSMWB\, P''$ for any $P''$ such that $P' \wt{} P''$. %\nts{(Think or xonfirm: on a universal path or an existential path?)}; % if $P'$ is on the divergent path.
\end{lemma}
\begin{proof}
Before going ahead, we provide some observation. If $P$ is not divergent, then the result obviously hold, because we can consume all the finite $\tau$ actions $P$ can fire to reach a state that can make no more $\tau$. That state satisfies the claim of the lemma (in a void way).
Now assume $P$ is divergent.
%For (1), 
\bctrmd{Intuitively}, the number $k$ is referring to those $\tau$'s that are not introduced by the replication. Consuming these $\tau$ actions leads $P$ to a state ready for starting the divergent path. 
By ``the divergent path'', we mean that the process can make the same $\tau$ over and again, without changing the states in absence of the restriction operation. 

Now we make the formal proof by induction on the structure of $P$.
%\fbox{\nts{TODO... need making precise?? More detailed discussion.}}

\begin{itemize}
% \item $P$ is $0$. Trivial.
% \item $P$ is $a.P$. Trivial.
% \item $P$ is $\overline{a}.P$. Trivial.
\item The cases when $P$ is $0$, $a.P$, or $\overline{a}.P$ are trivial.
\item $P$ is $P_1\para P_2$. There are several subcases.

\begin{enumerate}
\item $P_1$ is divergent while $P_2$ is not. \bctrmd{(For example, $P_1\DEF !a\para !\overline{a} \para \overline{c}$, $P_2\DEF c$. In here and what follows, we may use examples simply for more of an illustration, and they are not meant to be part of the proof anyhow.}) %\nts{More details: TODO?} 
% or indirectly divergent (\nts{TODO: need re-organize... }). 

In this case, the claim of the lemma holds for $P_1$ by induction hypothesis. That is, there exists $k_1\geqslant 0$ and $P_1'$ such that $P_1\st{\tau}_{k_1}\; P_1'$, and $P_1'\,\CCSMWB\, P_1''$ for any $P_1''$ such that $P_1' \wt{} P_1''$. Now since $P_2$ is not divergent, we can expire all the possible $\tau$ actions concerning $P_2$, including those by $P_2$ alone or from finite interactions between $P_2$ and $P_1$. By ``finite interactions'' we mean that the interactions are composed by (i.e., result from) visible actions that cannot be repeated for infinitely many times. In addition, notice that by means of interaction, we only eliminate (i.e., consume) those finite $\tau$ actions from inside $P_2$ or between $P_1$ and $P_2$. After this $P_2$ may still have infinite visible actions and  have infinite $\tau$'s with $P_1$, but this does not matter because they do not change the state of $P$, as there is an infinite repository of these actions (so the capacity of $P$ remains unchanged). %, enforcing the state to stay put.
 \ntstrmd{(OK?seems SO! CHECK!!)}
 Suppose this expiration operation constitutes $k_1'$ ($k_1'\geqslant 0$) $\tau$ actions. Then we know that $P\st{\tau}_{k_1+k_1'}\; P'$, and $P'\,\CCSMWB\, P''$ for any $P''$ such that $P' \wt{} P''$.

%Distinguish the case $P_1$ is directly divergent or indirectly divergent (i.e., it becomes divergent after having some finite interaction with $P_2$). In either case, the claim of the lemma holds. The former case is clear by induction hypothesis. The latter case, suppose after some finite interaction with $P_2$, $P_1,P_2$ become $P_1',P_2'$ respectively, where $P_1'$ is divergent. Then $P \wt{} P_1' \para P_2'$, and the result holds by ind.hyp. on $P_1'$.
%\nts{TODO more ??...}

\item $P_2$ is divergent while $P_1$ is not. Similar to the previous case.

\item Both $P_1$ and $P_2$ are divergent. \bctrmd{(For example, $P_1\DEF !a\para !\overline{a} \para \overline{c}$, $P_2\DEF c\para !b\para !\overline{b}$.)} %\nts{TODO}

By induction hypothesis, we know that there exists $k_i\geqslant 0$ ($i=1,2$) and $P_i'$ such that $P_i\st{\tau}_{k_i}\; P_i'$, and $P_i'\,\CCSMWB\, P_i''$ for any $P_i''$ such that $P_i' \wt{} P_i''$. Similar to case 1, we can remove those finite, say $k_3\geqslant 0$, $\tau$ actions between $P_1$ and $P_2$. Then we know that $P\st{\tau}_{k_1+k_2+k_3}\; P'$, and $P'\,\CCSMWB\, P''$ for any $P''$ such that $P' \wt{} P''$.

\item Neither of $P_1$ and $P_2$ is divergent, but $P$ is divergent. 
We break into two subcases.
\begin{enumerate}
\item 
We first separate a subcase that $P_1$ is indirectly divergent (the case for $P_2$ is similar). \\
\bctrmd{(For example, $P_1\DEF \overline{c}.(!a\para !\overline{a})$, $P_2\DEF c$)}\\ %
That is, $P_1$ becomes divergent after having some finite interaction with $P_2$. Suppose after some finite interaction with $P_2$, $P_1$ and $P_2$ become $P_1'$ and $P_2'$ respectively, where $P_1'$ is divergent. Then $P \wt{} P_1' \para P_2'$, and the result holds by proceeding as in case 1.

\item 
Excluding the indirectly divergent cases as in (a), we may suppose $P\equiv P_1\para P_2 \st{\tau}_{k} P_1'\para P_2'$ in which $k\geqslant 0$ and neither $P_1'$ nor $P_2'$ can do $\tau$ action. \\
\bctrmd{(For example, $P_1\DEF !a\para \overline{c}$, $P_2\DEF c\para !\overline{a}$)} \\ %
That is, we can deplete all the finite interactions in/between $P_1$ and $P_2$ 
%with the meaning of ``finite interaction" % 
as in case 1.
Now since $P$ is divergent, we have two observations: first, there must be some interaction between $P_1'$ and $P_2'$; second, the visible actions comprising the interaction, and thus the interaction itself, can recur forever. Thus suppose, in particular, that $P_1' \wt{\alpha} $ and $P_2' \wt{\bar{\alpha}} $, and these actions run infinitely. 
%The observations are not hard to verify, because otherwise $P$ would not be divergent. 
Then $P_1' \para P_2'$ is exactly the $P'$ as needed, because it holds that $P_1' \para P_2' \;\st{\tau}_{k'}\; P_1''\para P_2'' \,\CCSMWB\, P_1' \para P_2'$ for any $k'\geqslant 0$, % for any $P_1'',P_2''$ 
due to the infinite repository of the same actions. 
\ntstrmd{(OK?seems SO! CHECK!!)}  

\end{enumerate}
%\nts{TODO more analysis ?? maybe...}
\end{enumerate}

% \sepp
% \fbox{\nts{NEED more CHECK!! \rc{HERE!}} }
% \sepp

\item $P$ is $!P_1$. Suppose $!P_1$ can do some $\tau$ (the result trivially holds otherwise). Then taking $k$ as zero would be sufficient, by Lemma \ref{l:rep_invar}. That is, %we take $k$ to be $0$, and 
all the incoming $\tau$ actions do not change the state of $!P_1$.
\end{itemize}
Now the proof is completed.
\end{proof}

%-----------------------------------------------------------------
%\SEPALine
\sepp

% \sepp\sepp
% \fbox{\nts{\large \xxa{RE-RE-RE-ENGAGE FROM \rc{$\clubsuit\clubsuit$ HERE HERE HERE!!!}: 
% %(0) ADD the missing parts (e.g., Defs); 
% }}}
% \sepp\sepp

% \sepp\sepp
% \fbox{\nts{\large \xxa{CLEAR FROM \rc{$\clubsuit\clubsuit$ HERE HERE HERE!!!}: 
% %(0) ADD the missing parts (e.g., Defs); 
% }}}
% \sepp\sepp

By virtue of Lemma \ref{l:ccs_tau_state_div0}, the number of such $\tau$ actions as state-changing  is finite in a sense. 
That is to say, a \CCSm\ process %, from its start, 
can only make finite state-changing $\tau$ actions in a sequence, without any intertwining visible actions. 

\begin{corollary}%[\nts{TODO: Merge into Lemma \ref{l:ccs_tau_state_div}}]
\label{l:state-chg-bisi-prop0}  
%A \CCSm\ process %, from its start, 
%can only make finite state-changing $\tau$ actions in a sequence, without any interwining visible actions. 

Suppose $P$ can make a (possibly infinite) sequence of $\tau$ actions, then only a finite initial %preceding 
segment of it is state-changing. This $\tau$ sequence can be written as $P\st{\tau} P_1 \st{\tau} P_2 \cdots \st{\tau} P_k \st{\tau} \cdots$, where $P$ and $P_1$ through $P_k$ are state-changing and those afterwards (if any)  are state-preserving.
\end{corollary}
\begin{proof}
%\fbox{\nts{TODO: expannd and analyze ... \rc{HERE!}} }

%maybe it is a corollary of what we have proven (previously lemma11(1)?; now lemma 13?). seems Yes, just choose the minimal $k$ as designated by the previously lemma11(1) (now lemma 13? up-to change). 

This is a consequence of what we have proven in Lemma \ref{l:ccs_tau_state_div0}. One simply chooses the minimal $k$ as designated by that lemma.
\end{proof}

\noindent\textit{Remark.}~
%We remark about Lemma \ref{l:state-chg-bisi-prop0}. 
In general, it is not true that a process can only have finite state-changing $\tau$ actions (not necessarily from the start and with-in a sequence). For example, the process $!a.(b\para \overline{b}.c)$ can have infinite state-changing $\tau$ actions. 
However, the $\tau$ actions in a row without any in-between visible actions can only have finite state-changing (consecutive) $\tau$ actions as stated in the foregoing corollary. 

\sepp\sepp

Furthermore, a $\tau$ action that changes the state of a process, 
which may not come from a replication, 
should be simulated also by a state-changing $\tau$. 
We capture these ideas in the following two lemmas.

% \sepp\sepp
% \fbox{\nts{\large \xxa{RE-RE-RE-ENGAGE FROM \rc{$\clubsuit\clubsuit$ HERE HERE HERE!!!}: 
% %(0) ADD the missing parts (e.g., Defs); 
% }}}
% \sepp\sepp

\begin{lemma}\label{l:state-chg-bisi-prop1}
Assume $P\,\CCSMWB\, Q$. Then in the $\tau$ sequences that $P$ and $Q$ can make respectively (as described in Corollary \ref{l:state-chg-bisi-prop0}), the numbers of state-changing $\tau$ actions are equal.
\end{lemma}
\begin{proof}

By Corollary \ref{l:state-chg-bisi-prop0}, we can assume the $\tau$ sequences by $P$ and $Q$ are respectively 
\[
\begin{array}{l}
P\st{\tau} P_1 \st{\tau} P_2 \cdots \st{\tau} P_k \st{\tau} \cdots \\
Q\st{\tau} Q_1 \st{\tau} Q_2 \cdots \st{\tau} Q_{k'} \st{\tau} \cdots
\end{array}
\] where $P$ (respectively $Q$) and $P_1$ (respectively $Q_1$) through $P_k$ (respectively $Q_{k'}$) are state-changing and those afterwards (if any) are state-preserving. We need to show $k=k'$. Assume to the contrary that $k\neq k'$, say $k < k'$ (the case $k > k'$ is similar).

We then show that this would break the bisimulation between $P$ and $Q$. 
We observe that every $\tau$ action is composed of two complementary visible actions (since there is no restriction operation, every interaction is based on a visible name), so one can make these two visible actions in a sequential manner to reach the same state.
% (this sequence can be somehow considered atomic; though this is only imaginary, it is harmless to think this way informally). %and not formaly done, it is ) 
 Say $P\st{\tau} P'$ through interaction over $a$, then from the operational semantics we must have $P\st{a}\st{\overline{a}} P'$, which we call a pair-action.
Moreover, because the $\tau$ considered here is state-changing, these complementary visible actions must not be repeated forever (otherwise they would lead to repeated $\tau$ actions, making it state-preserving). So $P$ can reach $P_k$ in finite such pair-actions.
Since $P\,\CCSMWB\, Q$, it must be the case that $Q$ consumes exactly the same number of pair actions (of visible actions), and this corresponds to $k$ state-changing $\tau$ actions. That is, $Q$ must evolve to $Q_k$ in order to keep pace with $P$ so that the bisimulation can be maintained. Yet now as assumed,  $Q$ can make more state-changing $\tau$ actions, i.e., more pair-actions correspondingly. These (non-repeating) pair actions cannot be bisimulated by $P$ since $P$ at that stage only has state-preserving $\tau$ actions {(and consequently repeating pair actions)}. 
\ntstrmd{(last sentence OK?seems SO!)} 
%So we discern a contradiction.
So we have a contradiction.
% \sepp
% \fbox{\nts{NEED more CHECK!! \rc{HERE the last sentence!}} }
% \sepp
%\fbox{\nts{TODO: expannd and analyze with CHECK ... \rc{HERE!}} }
\end{proof}

%Below comes another lemma on internal actions that are state-changing.

% \sepp\sepp
% \fbox{\nts{\large \xxa{RE-RE-RE-ENGAGE FROM \rc{$\clubsuit\clubsuit$ HERE HERE HERE!!!}: 
% %(0) ADD the missing parts (e.g., Defs); 
% }}}
% \sepp\sepp

%It might be helpful to have some \bctrm{intuition} first: 
Lemma \ref{l:state-chg-bisi-prop1} helps to develop the following intuition: 
in order to simulate a state-changing $\tau$, one must do at least one $\tau$ so as to change its state; moreover, it cannot do more than one state-changing $\tau$ because it has only finite such $\tau$ actions in the (bisimulating) sequence of internal actions, even though it can do a couple of state-preserving $\tau$ actions; making such a state-changing $\tau$ is sufficient for the bisimulation.

Hence we have the follow-up lemma on how to match internal actions that are state-changing.

\begin{lemma}%[\nts{TODO: Merge into Lemma \ref{l:ccs_tau_state_div}}]
\label{l:state-chg-bisi-prop}
Suppose $P\,\CCSMWB\, Q$. If $P\st{\tau} P'$ in which the $\tau$ is state-changing, then $Q\st{\tau} Q' \,\CCSMWB\, P'$ in which the $\tau$ is state-changing as well.
\end{lemma}
\begin{proof}
% \sepp
% \fbox{\nts{NEED more CHECK!! \rc{HERE!}} }
% \sepp

%\fbox{\nts{TODO: expannd and analyze ... \rc{HERE!}} }
%\nts{TODO: expannd and analyze ...}
% Suppose $P\st{\tau} P'$ in which the $\tau$ is state-changing, since $P\,\CCSMWB\, Q$, we know that $Q\wt{} Q_1$ for some $Q_1$ and $P' \,\CCSMWB\, Q_1$. Because the $\tau$ by $P$ is state-changing, it must originate from two complementary visible actions, say $a$ and $\overline{a}$, that cannot be simulated by $P'$. Thus $Q$ must not do nothing, so by Corollary \ref{l:state-chg-bisi-prop0}, we can rewrite the simulation by $Q$ as $Q \st{\tau} Q_2 \wt{} Q_1 \,\CCSMWB\, P'$ in which $Q$ to $Q_2$ is state-changing.  %(\nts{Can we? Elaborate?}). 

% Now if the $\wt{}$ by $Q_2$ invloves only state-preserving $\tau$ actions, then we are done. If not, it then at least consumes one more state-changing $\tau$ action in the sequence. In this case, since $Q$ can only make finite state-changing $\tau$ action sequence (by Corollary \ref{l:state-chg-bisi-prop0}) whose length is the same as $P$ (by Lemma \ref{l:state-chg-bisi-prop1}), $Q$ will lose the pace with $P$, i.e., the step-wise state-changing bisimulation, and thus fail to bisimulate it eventually. Therefore, it must be the case that the $\tau$ actions from $Q_2$ to $Q_1$ are state-preserving, i.e., $Q_2 \,\CCSMWB\, Q_1 \,\CCSMWB\, P'$. That is, we have $Q\st{\tau} Q_2 \,\CCSMWB\, P'$, and $Q_2$ is the $Q'$ we want. \nts{(this paragraph OK?seems SO!)}

Suppose $P\st{\tau} P'$ in which the $\tau$ is state-changing, since $P\,\CCSMWB\, Q$, we know that $Q\wt{} Q_1$ for some $Q_1$ and $P' \,\CCSMWB\, Q_1$. Because the $\tau$ by $P$ is state-changing, it must originate from two complementary visible actions, say $a$ and $\overline{a}$, that cannot be simulated by $P'$ anyhow. Thus $Q$ must not do nothing, so we can rewrite the simulation by $Q$ as $Q \st{\tau} Q_2 \wt{} Q_1 \,\CCSMWB\, P'$. We claim that in this simulation, the $\tau$ in $Q \st{\tau} Q_2$ is state-changing, and the $\tau$ actions in $Q_2 \wt{} Q_1$ are state-preserving. Then the result of the lemma follows. %(\nts{Can we? Elaborate?}). 

To see why $Q$ is forced to simulate with exactly one state-changing $\tau$, we note that $Q$ can only make finite state-changing $\tau$ action sequence (by Corollary \ref{l:state-chg-bisi-prop0}) whose length is the same as that of $P$'s (by Lemma \ref{l:state-chg-bisi-prop1}). If the simulation consumes any more or fewer state-changing $\tau$ actions in the sequence, $Q$ will eventually lose the pace with $P$ and fail to bisimulate it. Thus $P$ and $Q$ are %necessarily 
obliged to engage a step-wise state-changing bisimulation, as they only hold a finite number %repository 
of such actions. 

Therefore, it must be the case that the $\tau$ in $Q \st{\tau} Q_2$ is state-changing, and the $\tau$ actions from $Q_2$ to $Q_1$ are state-preserving, i.e., $Q_2 \,\CCSMWB\, Q_1 \,\CCSMWB\, P'$. To recap, now we have $Q\st{\tau} Q_2 \,\CCSMWB\, P'$, and $Q_2$ is the $Q'$ we seek in the statement of the lemma. \ntstrmd{(this paragraph OK?seems SO!)}
\end{proof}

A corollary out of Lemma \ref{l:state-chg-bisi-prop1} and Lemma \ref{l:state-chg-bisi-prop} is that a state-preserving $\tau$ must be simulated by state-preserving one(s), since it cannot be bisimulated by a state-changing $\tau$ action.
\begin{corollary}\label{cor:state-chg-bisi-prop1}
Suppose $P\,\CCSMWB\, Q$. If $P\st{\tau} P'$ in which the $\tau$ is state-preserving, then $Q\st{\tau} Q' \,\CCSMWB\, P'$ in which the $\tau$ is state-preserving as well.
\end{corollary}

%-----------------------------------------------------------------
%\SEPALine
\sepp

The upcoming lemma somehow generalizes the results in the foregoing lemmas.

%We continue to examine the state-preserving $\tau$ actions in the coming up lemma.
\begin{lemma}\label{l:ccs_tau_state_div}
%If $P$ is not capable of an infite number of action (visible or silent) and $P\st{\tau} P'$, then $P \,\NCCSMWB\, P'$.
%Assume $P$ is divergent. Then 
Assume $P$ and $Q$ are \CCSm\ processes, and $k$ is as described in Lemma \ref{l:ccs_tau_state_div0} for $P$. If $P\CCSMWB Q$, then
\begin{enumerate}
%\item there exists $k\geqslant 0$ and $P'$ s.t. $P\st{\tau}_k\; P'$, and $P'\,\CCSMWB\, P''$ for any $P''$ s.t. $P' \wt{} P''$. %\nts{(Think or xonfirm: on a universal path or an existential path?)}; % if $P'$ is on the divergent path.
\item $Q\st{\tau}_k\; Q'$, and $Q'\,\CCSMWB\, Q''$ for any $Q''$ such that $Q' \wt{} Q''$;
\item there is a minimal $k$ satisfying both (1) for $Q$ and the property as described in Lemma \ref{l:ccs_tau_state_div0} for $P$; i.e., $P$ and $Q$ make the same number of state-changing $\tau$ action sequence before (synchronously) entering (state-preserving) divergence (if any).
\end{enumerate}
\end{lemma}
\begin{proof}
It should be clear that (2) follows from (1) and Lemma \ref{l:ccs_tau_state_div0}.

% \sepp
% \fbox{\nts{NEED more CHECK!! \rc{HERE!}} }
% \sepp

% \sepp\sepp
% \fbox{\nts{\large \xxa{RE-RE-RE-ENGAGE FROM \rc{$\clubsuit\clubsuit$ HERE HERE HERE!!!}: 
% %(0) ADD the missing parts (e.g., Defs); 
% }}}
% \sepp\sepp

For (1), one has to make sure that $Q$ has the same behavioural pattern. To see this, we remember that the bisimulation here requires synergistic divergence, so 
%the divergence of $Q$ is somewhat tied to that of $P$, i.e., 
$Q$ is divergent too whenever $P$ is.
We also notice that $Q$ has the property due to Lemma \ref{l:ccs_tau_state_div0}, i.e., there exists $k'\geqslant 0$ and $Q'$ such that $Q\st{\tau}_{k'}\; Q'$, and $Q'\,\CCSMWB\, Q''$ for any $Q''$ such that $Q' \wt{} Q''$. 
%We need to show that $k=k'$ (actually $k\geqslant k'$ and $k'\geqslant k$). Intuitively, if $k \neq k'$, then we can derive that $P$ and $Q$ are not bisimilar. \nts{how??}  

For the sake of convenience and w.l.o.g., we assume that $k$ and $k'$ are both the minimal such integers respectively (the result for `not minimal' follows from this case).  
%We need to show that $k\geqslant k'$. 
We need to show that $k = k'$. 
Intuitively, if $k < k'$ or $k > k'$, then we can derive that $P$ and $Q$ are not bisimilar. 
Assume to the contrary that $k < k'$ (the case $k > k'$ is similar). This means somehow $Q$ becomes state-preserving, say divergent, later than $P$. 
In light of Corollary \ref{cor:state_change}, starting from $P\CCSMWB Q$ and \bctrmd{before entering the state-preserving stage (diverging if any)}, both $P$ and $Q$ change their states stepwise in each `bisimulating' action. %stage. ~~~ 
%\nts{what now?? }  
%{\normalsize
Specifically, % / Intuitively, 
if $k < k'$, consider the following two sequences. 
\[
\begin{array}{l}
P\equiv P_0 \st{\tau} P_1 \st{\tau} P_2 \cdots \st{\tau} P_k \st{\tau} P_{k+1} \\\\
Q\equiv Q_0 \st{\tau} Q_1 \st{\tau} Q_2 \cdots \st{\tau} Q_k \st{\tau} Q_{k+1} \st{\tau} \cdots \st{\tau} Q_{k'} 
\end{array}
\]

We have %several observations.
\begin{enumerate}
\item[(1)] $P_i\NCCSMWB P_{i+1}$ ($i=0,...,k-1$). Before jumping into state-preserving stage (i.e., divergence), % if any), 
every $\tau$ changes the state of the process initiating from $P$ 
(recall that we choose the minimal $k$) %satisfying the first property (1) of this lemma).
%\fbox{\nts{NEED more CHECK of this observation!! \rc{HERE!}} }

\item[(2)] $P_k\CCSMWB P_{k+1}$. Once entering the phase of state-preserving, % (divergence), 
no $\tau$ actions changes the state.
%\fbox{\nts{NEED more CHECK of this observation!! \rc{HERE!}} }

\item[(3)] $Q_i\NCCSMWB Q_{i+1}$ ($i=0,...,k'-1$). For a reason similar to  (1). 
%\fbox{\nts{NEED more CHECK of this observation!! \rc{HERE!}} }

\item[(4)] $P_i\CCSMWB Q_i$ ($i=0,...,k$). This, % is by arrangement, cna 
can be ensured by Lemma \ref{l:state-chg-bisi-prop} and Lemma \ref{l:state-chg-bisi-prop1}, since the $\tau$ actions (until $P_k$ and $Q_k$) change the states.
%This follows from $P\CCSMWB Q$, because otherwise $P$ would fail to bisimulate $Q$ \nts{???right???NEED more CHECK here!?!?!}). \fbox{\nts{NEED more CHECK of this observation!! \rc{HERE!}} }

\end{enumerate}

Now these observations lead to a contradiction, i.e., we have $P_k \CCSMWB P_{k+1} \CCSMWB Q_{k+1}$ since $P_k$ must simulate $Q_k$'s $\tau$, but then $Q_k \CCSMWB Q_{k+1}$, which is contradictory. 
This scenario is depicted as in Fig. \ref{fig:proof4lemma21}, in which the place where the contradiction emerges is highlighted with the bold font. That is, the transitive path $Q_k\,\CCSMWB\, P_k \,\CCSMWB\, P_{k+1} \,\CCSMWB\, Q_{k+1}$ yields a contradiction with the edge connecting $Q_k$ and $Q_{k+1}$, i.e., $Q_k \,\NCCSMWB\, Q_{k+1}$.

%}%normalize
%\fbox{\nts{TODO...more check/expanding ??}} %... (use Corollary \ref{cor:state_change} to prove (2)?)}}\\
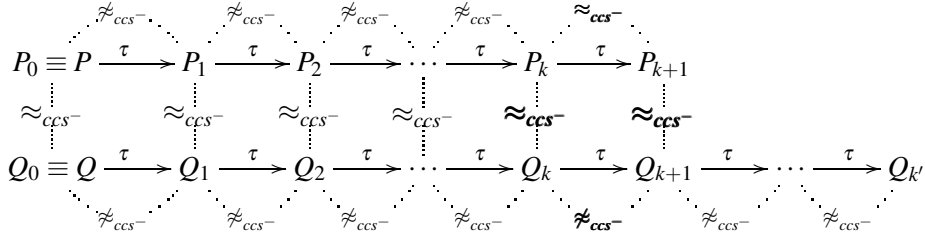
\begin{figure*}[!htb]
\centering
  % \subfigure[]{
  %   \includegraphics[width=6cm,height=2cm]{photo.eps}}
  % \subfigure[]{
  %   \includegraphics[width=6cm,height=2cm]{photo.eps}}
\[\hspace*{1cm} %NEED: \usepackage{etex}
\xymatrix@C=25pt{
P_0\equiv P \ar@/^1.6pc/@{.}[0,1]|{\NCCSMWB} \ar@{.}[d]|-{\displaystyle \CCSMWB} \ar[r]^{\tau} 
& P_1 \ar@/^1.6pc/@{.}[0,1]|{\NCCSMWB} \ar[r]^{\tau} \ar@{.}[d]|-{\displaystyle \CCSMWB} & P_2 \ar@/^1.6pc/@{.}[0,1]|{\NCCSMWB} \ar[r]^{\tau} \ar@{.}[d]|-{\displaystyle \CCSMWB} &  \cdots \ar@/^1.6pc/@{.}[0,1]|{\NCCSMWB} \ar[r]^{\tau} \ar@{.}[d]|-{\displaystyle \CCSMWB} &  P_k \ar@/^1.6pc/@{.}[0,1]|{\pmb{\CCSMWB}} \ar[r]^{\tau} \ar@{.}[d]|-{\displaystyle \pmb{\CCSMWB}} & P_{k+1} \ar@{.}[d]|-{\displaystyle \pmb{\CCSMWB}} & & & \\
%r\ar@{}[r]|-{\displaystyle \leq} \ar[d]^{\widehat{\mu}} & \app { \encom N} r \ar@{=>}[d]^{\widehat{\mu}} \\
%S \ar@{}[r]|-{\displaystyle \geq}& C_1[\encom{M_1}] & &  & C_1[\encom{N_1}]\ar@{}[r]|-{\displaystyle \leq} & T
Q_0\equiv Q \ar@/_1.6pc/@{.}[0,1]|{\NCCSMWB} \ar[r]^{\tau}
& Q_1 \ar@/_1.6pc/@{.}[0,1]|{\NCCSMWB} \ar[r]^{\tau} & Q_2 \ar@/_1.6pc/@{.}[0,1]|{\NCCSMWB} \ar[r]^{\tau} &  \cdots \ar@/_1.6pc/@{.}[0,1]|{\NCCSMWB} \ar[r]^{\tau} &  Q_k \ar@/_1.6pc/@{.}[0,1]|{\pmb{\NCCSMWB}} \ar[r]^{\tau} & Q_{k+1} \ar@/_1.6pc/@{.}[0,1]|{\NCCSMWB} \ar[r]^{\tau} & \cdots \ar@/_1.6pc/@{.}[0,1]|{\NCCSMWB} \ar[r]^{\tau} &  Q_{k'}\\
}
\]  
  \caption{Figure for the proof of Lemma \ref{l:ccs_tau_state_div}.} \label{fig:proof4lemma21}
\end{figure*}
\end{proof}

%-----------------------------------------------------------------
%\SEPALine
\sepp

Next comes a crucial property for analyzing the matching of $\tau$ actions in the sense of the weak bisimulation equality. %\nts{Are we moving too fast?}

% \sepp\sepp
% \fbox{\nts{\large \xxa{RE-RE-RE-ENGAGE FROM \rc{$\clubsuit\clubsuit$ HERE HERE HERE!!!}: 
% %(0) ADD the missing parts (e.g., Defs); 
% }}}
% \sepp\sepp

\begin{proposition}\label{p:tau_case}
Assume $P\CCSMWB Q$ and $P\st{\tau} P'$, then $Q \,\st{\tau} \CCSMWB\, P'$.
\end{proposition}
\begin{proof}
By the assumption, we know that $P\st{\tau} P'$ implies $Q \,\wt{} \CCSMWB\, P'$. We need to more precisely analyze the matching weak transition by $Q$ (in absence of the restriction operation), with the help of 
%Lemma \ref{l:ccs_tau_state} through Lemma \ref{l:ccs_tau_state_div} and the accompanying corollaries.
the lemmas thus far and the accompanying corollaries.

The analysis separates the cases whether $P$ is divergent or not. We first tackle the case when $P$ is not divergent. By Lemma \ref{l:ccs_tau_state}, we know that $P \,\NCCSMWB\, P'$, i.e., that $\tau$ is state-changing. Then By Lemma \ref{l:state-chg-bisi-prop}, we have $Q\st{\tau} Q' \,\CCSMWB\, P'$ in which $\tau$ is state-changing.

% \sepp
% \fbox{\nts{NEED more CHECK!! or STRENGTHEN using Lemmas \ref{l:ccs_tau_state}---\ref{l:ccs_tau_state_div}~~ \rc{HERE!}} }
% \sepp

% \sepp
% \fbox{\nts{NEED more CHECK!! \rc{HERE!}} }
% \sepp

% A straightfoward observation is that $Q$ must not simulate $P$ by null move, otherwise we have $P \,\CCSMWB\, Q \,\CCSMWB\, P'$, contradicting Lemma \ref{l:ccs_tau_state}. Therefore, it must be the case that $Q\st{\tau} Q' \wt{} Q'' \CCSMWB P'$. We now claim that $Q'\CCSMWB Q''$, thus completing the proof. Actually, we prove something stronger, i.e., $Q'\CCSMSB Q''$.  
% \nts{......}

Now we consider the case when $P$ is divergent. Since $P\CCSMWB Q$, we know that $Q$ diverges too. 
Assume $k$ is as decided by %Lemmas \ref{l:ccs_tau_state_div0} and 
Lemma \ref{l:ccs_tau_state_div}. That is, $P$ and $Q$ make the same number %$k$ 
of state-changing $\tau$ action sequence before (simultaneously) entering divergence (in which $\tau$ is state-preserving). 
There are two subcases.
\begin{enumerate}
\item[(1)] If $k$ equals $0$, then the $\tau$ in $P \st{\tau} P'$ must be state-preserving. By Corollary \ref{cor:state-chg-bisi-prop1}, $Q\st{\tau} Q' \,\CCSMWB\, P'$ in which $\tau$ is state-preserving.

\item[(2)] If $k$ is not $0$, then the $\tau$ in $P \st{\tau} P'$ can be state-changing or state-preserving (i.e., from the divergence after $k$). In the former, the result follows as above in the case $P$ is not divergent. In the latter, we conclude as in (1).

%\item[] \ntstrmd{(the above two-case argument OK?seems SO!)}
\end{enumerate}
\vspace*{-.5cm}
%\rc{In light of that Lemma, if $k$ equals $0$ in that lemma, then $Q\st{\tau} \CCSMWB Q'$ and therefore $Q'\CCSMWB P'$ because $P\CCSMWB P'$.}
%\fbox{\nts{TO CEHCK \rc{HERE}...}}

% Otherwise if $k$ is not zero, then by resorting to (2) of Lemma \ref{l:ccs_tau_state_div}, we know the result also holds (actually the proof of (2)(on top of (1) and Lemma \ref{l:ccs_tau_state_div0}) of that lemma entails the result here). 
% This is because at the stage before entering divergence (i.e., having exactly $k$ state-switching $\tau$ actions), each $P$'s $\tau$ must be bisimulated by one state-switching $\tau$ of $Q$, otherwise they would not be bisimilar after consuming all $k$ number of $\tau$ actions (reaching the ready-for-diverging state). So we are done.
% \fbox{\nts{TO CEHCK \rc{HERE}...(notice mutual simulation??)}}
% %\nts{(??)}~\nts{more CHECK ...(notice mutual simulation??)}%\fbox{\nts{CHECK/DO: mutual simulation??} }
\end{proof}

%\sepp
%-----------------------------------------------------------------
%\SEPALine
\sepp

Having done analyzing the $\tau$ action in a bisimulation, we finally deal with the situation for visible actions.
\begin{proposition}\label{p:visible_case}
% Assume $P\CCSMWB Q$ and $P\st{\alpha} P'$ in which $\alpha$ is not $\tau$, then 
% %\cancel{$Q \,\st{\alpha} \CCSMWB\, P'$} ~~~ 
% $Q \,\wt{} Q_1\st{\alpha} Q_2 \CCSMWB\, P'$, in which $Q \CCSMWB Q_1$, for some $Q_1, Q_2$.
Assume $P\CCSMWB Q$ and $P\st{\alpha} P'$ in which $\alpha$ is not $\tau$, then 
$Q \,\wt{} Q_1\st{\alpha} Q_2 \CCSMWB\, P'$ for some $Q_1, Q_2$, in which $Q \CCSMWB Q_1$ and the $\tau$ actions in $Q \,\wt{} Q_1$ are state-preserving.
\end{proposition}
\begin{proof}
By the premise, $Q$ must simulate by $Q\wt{} Q_1 \st{\alpha} Q_2 \wt{} Q' \,\CCSMWB\, P'$ for some $Q_1, Q_2$, and $Q'$. 
%Particularly, we need to show that $Q_2$ is forced to have no internal actions in bisimulating $P$. 
%Otherwise the bisimulation between $P$ and $Q$ breaks down. %We proceed by contraposition. 
% \sepp
% \fbox{\nts{NEED more ANALYSIS / CHECK!! \rc{HERE!}} }
% \sepp
%\bc{(Intuitively)} 
In this simulation, by Lemmas \ref{l:state-chg-bisi-prop1}, \ref{l:state-chg-bisi-prop} and Corollary \ref{cor:state-chg-bisi-prop1}, %it can be inferred that   
%$Q\wt{} Q_1 \st{\alpha} Q_2 \wt{} Q'$, 
the $\tau$ sequence between $Q$ and $Q_1$, and the $\tau$ sequence between $Q_2$ and $Q'$, contain no state-changing $\tau$ actions. If assumed otherwise, $Q'$ would be unable to match $P'$ due to being short of enough state-changing $\tau$ actions as compared with those of $P'$. So the internal action sequences in the bisimulation must be state-preserving.  
Consequently, we have $Q\wt{}\st{\alpha} Q_2 \,\CCSMWB\, P'$, as needed. 
% \fbox{
% \begin{minipage}{12cm}
% \rc{\Large STOP!STOP!STOP! WRONG here!}
% \nts{
% Say $P_1\DEF !c.d \para !\overline{c} \para d$ and $P_2\DEF !c.d \para !\overline{c} \para !c$. Obviously $P_1 \,\NCCSMSB\, P_2$. However $P_1 \,\CCSMWB\, P_2$, because the action 
% \[P_1\,\st{d}\, !c.d \para !\overline{c}
% \] can be simulated by
% \[P_2\,\st{\tau}\, !c.d \para !\overline{c} d\para 0 \para!c \,\st{d}\, !c.d \para !\overline{c} 0\para 0 \para!c
% \] Henceforth, every $d$ produced by $P_1$ in simulating the subprocess $!c$ in $P_2$ can be simulated in a similar way.
% } 
% \end{minipage}
% }
\end{proof}
In general, the result stated in Proposition \ref{p:visible_case} turns out to be the best we can do to strengthen the simulation of a visible action, as opposed to the counterexample  that distinguishes the strong and weak bisimilarities (see Lemma \ref{con:ccs_bisi_ws_coin}). %((slightly maybe)). 

Nonetheless, %as a matter of fact, 
Proposition \ref{p:visible_case} can be refined further if one simply focuses on non-divergent processes. In that case, we end up with %somewhat jump onto 
the strong bisimilarity, as Corollary \ref{cor:visible_case} reveals. 
With this corollary in position, it makes sense to speculate that, if the replication operator were to be eliminated from \CCSm\ (though this makes the calculus less interesting), then the weak bisimilarity would flat onto the strong bisimilarity. 
\begin{corollary}\label{cor:visible_case}
Assume $P$ is not divergent. If $P\CCSMWB Q$ and $P\st{\alpha} P'$ where $\alpha$ is not $\tau$, then $Q \,\st{\alpha} \CCSMWB\, P'$.
\end{corollary}
\begin{proof}

Since $P$ is not divergent, neither is $Q$. 
By Proposition \ref{p:visible_case}, we know that $Q\wt{} Q_1 \st{\alpha} Q_2 \,\CCSMWB\, P'$, where $Q\wt{} Q_1$ has only state-preserving $\tau$ actions. However, $Q$ is not divergent. This means that the $\tau$ actions in $Q\wt{} Q_1$ can only be state-changing. This is a contradiction, which leads to the only possibility that $Q\wt{} Q_1$ contains zero $\tau$ actions, i.e., $Q_1$ is $Q$. Thus we are done.
\end{proof}

Now, Proposition \ref{p:tau_case} and Proposition \ref{p:visible_case} amount to Theorem \ref{con:ccsm_bisi_ws_coin}, the main result.

\begin{proof}[Proof of Theorem \ref{con:ccsm_bisi_ws_coin}]
We define 
\[\R \DEF \{(P,Q) \,|\, P \,\CCSMWB\, Q\}.
\]
We show that \R is a 
%~~~ \cancel{strong bisimulation} ~~~ 
quasi-strong bisimulation. Assume $P\,\R\, Q$. We have two cases.
\begin{itemize}
\item $P\st{\alpha} P'$ in which $\alpha$ is not $\tau$. %\nts{ ... }
By Proposition \ref{p:visible_case}, 
%~~~\cancel{$Q \,\st{\alpha} Q' \CCSMWB\, P'$} ~~~ 
$Q \,\wt{}\st{\alpha} Q' \CCSMWB\, P'$. 
So we have $P'\,\R\, Q'$.

\item $P\st{\tau} P'$. %\nts{ ... }
By Proposition \ref{p:tau_case}, $Q \,\st{\tau} Q' \CCSMWB\, P'$. So we have $P'\,\R\, Q'$.
\end{itemize}\vspace*{-.5cm}
\end{proof}

%\sepp

%\paragraph*{\textit{Discussion}}~ 
\subsection{Further results and discussion} 
%
%We make some discussion about the result in this section.

%\item[] 
We argue that Corollary \ref{cor:coin-weak-quasi-strong} also holds for the branching bisimilarity \cite{GW89a,GW96} (in place of the weak bisimilarity). %The technical routine is more or less the same. 
This, in turn, would lead to the coincidence between the weak bisimilarity and the branching bisimilarity. It is an intriguing and practicable work that we now solidify.

First of all, it is helpful to exploit further Theorem \ref{con:ccsm_bisi_ws_coin}, particularly the relationship with the branching bisimilarity \cite{GW89a,GW96}, a well-known equivalence relation on processes that preserves the branching structure of processes. The definition of branching bisimulation is as follows. %To be consistent with the current setting, we impose divergence-sensitiveness on the bisimulation.
\begin{definition}\label{d:branching-bisi}
A symmetric binary relation \R on \CCSm\ processes is a branching bisimulation if it is divergence-sensitive, and whenever $P \,\R\, Q$, the following properties hold. 
\begin{itemize}
\item if $P\st{\alpha} P'$, then either 
\begin{itemize}
\item $\alpha$ is $\tau$ and $P' \,\R\, Q$; Or
\item $Q\wt{}Q''\st{\alpha} Q'$, $P\,\R\, Q''$ and $P'\,\R\, Q'$.
\end{itemize}
\end{itemize}
Two processes $P$ and $Q$ are branching bisimilar, notation $P \,\CCSMBB\, Q$, if there exists some branching bisimulation \R such that $P\,\R\, Q$.    
\end{definition}
We call $\CCSMBB$ the branching bisimilarity. To be consistent with the current setting, we impose divergence-sensitiveness on the branching bisimulation.   
It is not hard to see that the branching bisimilarity implies the weak bisimilarity, i.e., $\CCSMBB \,\subseteq\, \CCSMWB$.

Here comes an important observation of the proof of Theorem \ref{con:ccsm_bisi_ws_coin}. It uses Proposition \ref{p:visible_case}, which states that the internal actions in $Q \,\wt{}\st{\alpha} Q' \CCSMWB\, P'$ in the first clause of the proof of Theorem \ref{con:ccsm_bisi_ws_coin} are actually state-preserving. Following this observation, we can strengthen the definition of quasi-strong bisimulation without changing any distinguishing power. 
\begin{definition}\label{d:quasi-strong-branching-bisi}
%TOADAPT(adapt from the definition of weak bisimulation)
A symmetric binary relation \R on \CCSm\ processes is a quasi-strong branching bisimulation if it is divergence-sensitive, and whenever $P \,\R\, Q$, the following properties hold. 
\begin{itemize}
\item if $P\st{\alpha} P'$ and $\alpha$ is not $\tau$, then $Q\,\wt{}Q''\st{\alpha} Q'$, $P\,\R\, Q''$ and $P'\,\R\, Q'$.
\item if $P\st{\tau} P'$, then $Q\st{\tau} Q'$ and $P'\R Q'$.
\end{itemize}
Two processes $P$ and $Q$ are quasi-strongly branching bisimilar, notation $P \,\CCSMQSBV\, Q$, if there exists some quasi-strong branching bisimulation \R such that $P\,\R\, Q$.
\end{definition}

Through the same proof routine as that of Theorem \ref{con:ccsm_bisi_ws_coin}, we can infer that $\CCSMQSBV$ also coincides with the weak bisimilarity.
\begin{lemma}\label{l:coin-weak-quasi-strong-v}
In \CCSm, it holds that $\CCSMWB \;=\; \CCSMQSBV$.
\end{lemma}

Now examining the difference between the quasi-strongly branching bisimulation and the branching bisimulation, it is straightforward to see that both of the clauses of the quasi-strongly branching bisimulation implies that of the branching bisimulation. Hence the following lemma.
\begin{lemma}\label{l:coin-weak-quasi-strong-v-branching}
In \CCSm, it holds that $\CCSMQSBV \;\subseteq\; \CCSMBB$.
\end{lemma}
\begin{proof}
To show that $\CCSMQSBV$ is a branching bisimulation, we focus on the $\tau$-clause, since it is the only distinct part. 
In particular, the second clause of the quasi-strongly branching bisimulation implies that of the branching bisimulation because we can rewrite $Q\st{\tau} Q'$ as $Q\,\wt{}Q\st{\tau} Q'$.
\end{proof}

The lemmas above lead to the follow-up corollary.
\begin{corollary}\label{l:coin-weak-branching}
In \CCSm, it holds that $\CCSMWB \;=\; \CCSMBB$.
\end{corollary}
\begin{proof}
By Lemma \ref{l:coin-weak-quasi-strong-v-branching}, we have $\CCSMQSBV \;\subseteq\; \CCSMBB \,\subseteq\, \CCSMWB$. Then the equality follows by Lemma \ref{l:coin-weak-quasi-strong-v}.
\end{proof}

To conclude, all the discussion so far boils down to the next theorem.
\begin{theorem}\label{t:all-in-all-coin}
In \CCSm, it holds that $\CCSMWB \;=\; \CCSMQSB \;=\; \CCSMQSBV \;=\; \CCSMBB$.
\end{theorem}

\sepp

We make some more remarks before ending this section.
% \sepp\sepp
% \fbox{\nts{\large \xxa{RE-RE-RE-ENGAGE FROM \rc{$\clubsuit\clubsuit$ HERE HERE HERE!!!}: 
% %(0) ADD the missing parts (e.g., Defs); 
% }}}
% \sepp\sepp
%\begin{itemize}
%\item[] 
As mentioned, in spirit of Corollary \ref{cor:visible_case} and Proposition \ref{p:visible_case}, if \CCSm\ is further deprived of the replication, we believe that the weak bisimilarity would fall onto the strong bisimilarity. This is virtually not hard to verify, by means of going through all the analysis above but ignoring those parts concerning the replication.  We do not extend the discussion into details here, as a calculus with neither restriction nor replication appears not very interesting. 
%Although this appears a bit unsurprising from intuition, 
Nevertheless, it still sheds light on the essential gap between the weak and strong bisimulation equalities. It might be intriguing to see if one can find a subcalculus of CCS in which the weak and strong bisimilarities coincide, e.g., with a very special form of replication.
Also to this point, the analysis and result on quasi-strong bisimilarity may offer some potential tool for analyzing finite-state processes in a broader field. % (?).

%These 

%\item[] 
Conceivably, we can coin the quasi-strong bisimilarity in the higher-order model, in roughly the same vein as that of \CCSm. %One might think that in the higher-order setting, it is likely to obtain 
It is worthwhile to investigate whether the same coincidence exists between the quasi-strong bisimilarity and the weak bisimilarity. Due to the difference in communication machinery, the technical approach might be strikingly different.

\sepp

\iftoggle{hidingON}{%
  % use hiding
} {%
  % no hiding
{\tiny \rc{INCORRECT arguments here! TO REMOVE (or with little chance TO AMEND)!!}
One might think that in the higher-order setting, it is likely to obtain the same coincidence exists between the quasi-strong bisimilarity and the weak bisimilarity. 
Unfortunately this is not as expected. We provide a counterexample below.
\begin{itemize}
% \item A counterexample in \HOCCSm\ for the non-coincidence between weak bisimilarity and quasi-strong bisimilarity.
% \[P\DEF Q_d\para \overline{d}Q_d \qquad
% \mbox{ in which }
% Q_d\DEF d(X).e.(X\para \overline{d}X)
% \]
% Now $P\st{\tau} e.(Q_d\para \overline{d}Q_d) \DEF P_1$. 
% We have $P \,\HOMWB\, P_1$, but $P \,\NHOMQSB\, P_1$, because $P_1$ can make a $\tau$ whereas $P$ cannot. 
\item In \HOCCSm, a counterexample articulating the non-coincidence between the weak bisimilarity and the quasi-strong bisimilarity is as follows. Let $\HOCCSMQSB$ denote the quasi-strong bisimilarity for \HOCCSm, which is defined as Definition \ref{d:quasi-strong-bisi} except that for output the simulation is closed in the manner as the $(*)$ equation in Definition \ref{context-bisimulation}. Now we define the following process.
\[P\DEF Q_d\para \overline{d}Q_d \qquad
\mbox{ in which }\quad
Q_d\DEF d(X).e.(X\para \overline{d}X)
\]
Then $P\st{\tau} e.(Q_d\para \overline{d}Q_d) \DEF P_1$. 
We have \rc{\xcancel{$P \,\HOCCSMWB\, P_1$, but $P \,\NHOCCSMQSB\, P_1$}} (\rc{HERE $P$ and $P_1$ are NOT weakly bisimilar, because $P$ can make an output on $d$ while $P_1$ cannot.}), because $P$ can make a $\tau$ whereas $P_1$ cannot. 
\item The separation result between \HOCCSm\ and \CCSm\ concerning the existence of the quasi-strong bisimilarity, with the accompanying examples and discussion, inherently exhibits the essential difference between higher-order and first-order communications, even in the setting without the restriction operation. 
\end{itemize}
}
%\item 
}%hiding toggle

%\end{itemize}

%---------------------------
% Local Variables:
% mode: LaTeX
% TeX-master: "main.tex"
% End:
%\clearpage
%\input{part_pim.tex}\clearpage %this part to be moved, for future work. (renamed to "part_pim_NOTusedNOW.tex")
%% !TEX root = ./main.tex

\section{Conclusion}\label{s:conclusion}

% \sepp\sepp
% \fbox{\nts{\large \xxa{RE-RE-RE-ENGAGE FROM \rc{$\clubsuit\clubsuit$ HERE HERE HERE!!!}: 
% %(0) ADD the missing parts (e.g., Defs); 
% }}}
% \sepp\sepp

This paper has been focusing on the relationship between the strong and the weak bisimilarities, in process models from which the restriction operator is removed. We have presented a few observations about such relationship in both a first-order model (\CCSm) and a higher-order one (\HOCCSm). Basically, it is shown invariant in both models that the weak bisimilarity remains strictly weaker than the strong bisimilarity, even without the capacity of hiding information. Essentially, this is a consequence of the replication operation, though the situation is a bit different in \CCSm\ and \HOCCSm, because the replication is primitive in the former but derivable in the latter. Anyhow, what we have shown illustrates that the replication operation somehow can also `hide' information, but in a sharply different way, i.e., it offers plenty of the same processes (and actions). Though slightly beyond expectation, we can still succeed in reducing the distance between the strong and weak bisimilarities in \CCSm. That is, we show that in \CCSm, the strong bisimilarity can be approached by the so-called quasi-strong bisimilarity. Formally, this bisimilarity intensifies the weak bisimilarity in two respects: requiring strong bisimulation for silent actions and relinquishing trailing silent actions in matching a visible action. As it appears, the quasi-strong bisimilarity tightens up the weak bisimilarity, toward the strong bisimilarity. To this end, a key result is that the quasi-strong bisimilarity coincides with the weak bisimilarity. 
This coincidence conveys that the weak bisimilarity can be reinforced by demanding more than its original requirement about simulation, and thus becomes quasi-strong, while maintaining its original discriminating power. Moreover, it reveals that the absence of the restriction operator does not cause zero effect, and we can hopefully take advantage of this to make the weak bisimilarity more tractable. 
%Yet in \HOCCSm, this is again not the case, i.e., quasi-strong bisimilarity sits right inside weak bisimilarity. 
We have also discussed the relationship between the quasi-strong bisimilarity and the branching bisimilarity. As a significant spinoff, we show that the branching bisimilarity is coincident with the weak bisimilarity.

There are some questions worthy of further investigation. A first one is to prove or disprove the %Conjecture \ref{conj:hoccs_struc_striin_strong}. 
conjecture, aforementioned in Section \ref{s:hoccsm}, that in \HOCCSm\ the strong context bisimilarity collapses onto the structural congruence.
A second one is to seek in \HOCCSm\ a counterpart of the quasi-strong bisimilarity like Definition \ref{d:quasi-strong-bisi} for \CCSm\ (potentially with \emph{novel} idea), and attempt to prove its agreement with the weak context bisimilarity. Intuitively, we believe that the weak bisimilarity can be tightened toward the strong bisimilarity in \HOCCSm. This quest would become even more interesting, especially when taking into account the outcome of %Conjecture \ref{conj:hoccs_struc_striin_strong} 
the conjecture on the coincidence between the strong context bisimilarity and the structural congruence. % in the meanwhile. 
%A third one is to examine the similar results for the branching bisimilarity, in both \CCSm\ and \HOCCSm. 
One more direction is to exploit more (process) models without the restriction operator, e.g., value-passing models or (higher-order) ambient models, for properties that stem from the absence of the operator.

%---------------------------
% Local Variables:
% mode: LaTeX
% TeX-master: "main.tex"
% End:

%\input{bisimulation.tex}
%\input{bisimulation_originalContextVersion.tex}
%\input{bisimulation_varHOBisi1NotDoneYet.tex}
%\input{bisimulation_varHOBisi2NotDoneYet.tex}
%\input{ob_bisimulation.tex}
%\input{expressiveness.tex}
% \input{introduction.tex}
% \input{preliminary.tex}
% \input{encoding.tex}
% \input{normal.tex}
% \input{conclusion.tex}
%----------------------------------------------------------------------------------------------

%----------------------------------------------------------------------------------------------
%\sepp
%\noindent\textbf{Acknowledgements}\;\;
%%This work has been supported by project ANR 12IS02001 PACE and NSF of China (61261130589, 61472239, 61572318). 
%We are grateful to the referees of EXPRESS/SOS 2016 for their useful comments on this article and related work. We also thank Qiang Yin for the helpful discussion. 
%%Alan Schmitt
%\sepp

%-------------------------bibliography------------------------------------------------------------
%\clearpage
%\vspace*{-3mm}
%\nocite{*}
%\bibliographystyle{eptcs} %together with eptcs cls, this bib style will show doi<<<<<!!!!!
\bibliographystyle{eptcs}
\bibliography{process,ambients}
%-----------------------------------------------------------------------------------------------

%\bibliographystyle{unsrt}
%\bibliography{oo}

\end{document}